\documentclass{svjour3}                     
\usepackage[colorlinks=true,citecolor=blue]{hyperref}
\usepackage{graphicx}
\usepackage[margin=2.5cm]{geometry}
\usepackage{amsmath}
\usepackage{graphics} 
\usepackage{graphicx}  
\usepackage{epsfig} 
\usepackage{caption}
\usepackage{amsmath}
\usepackage{wrapfig}
\usepackage{graphicx}
\usepackage{subcaption}
\usepackage{amssymb}
\usepackage{natbib}

\usepackage{enumitem}
\usepackage[title]{appendix}
\begin{document}
\title{Mimicking the TYC strategy: Weak Allee effects, and a ``non" hyperbolic extinction boundary}	
\titlerunning{Mimicking the TYC strategy}        
\author{Eric M. Takyi \and Joydeb Bhattacharyya \and Matthew Beauregard \and Rana D. Parshad
}

\institute{Eric M. Takyi \at
	Department of Mathematics, Iowa State University, Ames, IA 50011, USA.
	\and
	Joydeb Bhattacharyya \at
	Department of Mathematics, Karimpur Pannadevi College, Nadia, WB 741152, India.
	\and
	Matthew Beauregard \at
	Department of Physics, Engineering, and Astronomy, Stephen F Austin State University, Nacogdoches, TX 75962, USA.
	\and
	Rana D. Parshad \at
	Department of Mathematics, Iowa State University, Ames, IA 50011, USA.}
	
\maketitle
\begin{abstract}
The Trojan Y Chromosome strategy (TYC)  is a genetic biocontrol strategy designed to alter the sex ratio of a target invasive population by reducing the number of females over time. Recently an alternative strategy is introduced, that mimics the TYC strategy by harvesting females whilst stocking males  \citep{L19} (FHMS). We consider the FHMS strategy, with a weak Allee effect,
and show that the extinction boundary need \emph{not} be hyperbolic. To the best of our knowledge, this is the first example of a non-hyperbolic extinction boundary in mating models, structured by sex.
Next, we consider the spatially explicit model and show that the weak Allee effect is \emph{both} sufficient and necessary for Turing patterns to occur. We discuss the applicability of our results to large scale biocontrol, as well as compare and contrast our results to the case with a strong Allee effect.
Correspondence to: rparshad@iastate.edu
\end{abstract}
\section{Introduction}

Invasive aquatic species are an imminent threat to marine biodiversity. The rate of invasions due to alien species worldwide continues to rise \citep{Haveletal}. Due to the various harmful effects of these species, their control is a paramount issue in ecology. Invasive species, upon effectively building up in another condition, can be hard to manage and the control expenses can get extreme. Gutierrez and Teem \citep{GutierrezTeem06} proposed an autocidal biocontrol TYC strategy to wipe out invasive species with $XX/XY$ sex chromosomes via a constant release of $YY$ males referred to as supermales. Exogenous sex hormones are utilized broadly to control sex in the aquaculture fishes. Male fish exposed to certain sex hormones can become feminized \citep{Set89}. Mating of a $XY$ phenotypic female fish and a wild-type $XY$ male fish produce supermale fish-bearing two $Y$ chromosomes. $YY$ supermales crossing to $XX$ females yield all $XY$ male offspring. The production of $YY$ broodstock of Nile tilapia \citep{Met97,V99}, yellow catfish \citep{Let13}, and brook trout \citep{Set16} has been proven successful by using TYC strategy. Further, for the production of  $YY$ supermales, supplying feminized $YY$ supermales into an undesired population was proposed by Gutierrez and Teem \citep{GutierrezTeem06}. There is a large literature on the TYC strategy

\citep{Parshad10,P11,Parshad13,TGP13,WP14,SDE2013,Z12,PG11,GP12}. Essentially,
\begin{itemize}
	\item For the classical four species TYC model, given any initial condition for the invasive wild-type males and females, there exist initial conditions for the introduced feminised supermales, and a threshold introduction rate, such that for an introduction greater than this rate, extinction of the wild-type occurs \citep{GP12,WP14,PG11}.

	\item The recent seminal work of Schill and collaborators \citep{Set16, Schill17, Schill18}, makes it evident that biocontrol of TYC type rests purely on the introduction of the supermale - \emph{not} feminized supermales, as these are still not in existence (and certainly not in mass production). Thus the three species TYC, with a one-time introduction of supermales, via the initial condition, is what occurs/is occurring in practice \citep{Schill18}.

	\item The three species TYC has been investigated in detail mathematically. The literature is rife with several results on well-posedness, and the long time dynamics of the system, under the assumptions of positive solutions (solutions that remain positive if they start from positive initial data) \citep{Parshad13,P11}. Essentially, here again, a sufficient introduction of the supermale can always yield extinction. However, the three species TYC model is now known to be ill-posed - and solutions to the female component can blow-up in finite time \citep{PBT19}, if the introduction of supermales is too large.

	\item In order to circumvent the issues of blow-up or ``unphysical" solutions, and due to the paucity of supermales, recent investigations into TYC type biocontrol have focused on (1) remodeling mating dynamics in TYC type models \citep{BPL20,JB20} and (2) investigation models that ``mimic" the TYC dynamics without using supermales such as using selective harvesting strategically \citep{L19, L18}.

	\item The consequences of a strong Allee effect on TYC type dynamics have also been recently investigated \citep{BPL20,JB20}. However, the impact of a \emph{weak} Allee effect on the population dynamics in case of a TYC type strategy and/or a FHMS strategy is adopted, has not been investigated.
\end{itemize}

Harvesting in practice, is tricky to say the least. Although it has been used in invasive species management along with chemical and biological control measures, and can result in non-random removals of individuals from targeted populations \citep{Bet2011,M2020}. The potential selectivity of these methods therefore has strong ecological and evolutionary implications. Consequently, we suggest that Palkovacs et al.’s \citep{Palkovacs2018} framework could be applied to invasive species management. Indeed, harvest-driven trait changes in invasive species might induce unexpected and potentially counterproductive results that may not have been explicitly considered by ecosystem managers. The work of Lyu \citep{L18,L19}, demonstrates the potential for harvest as an effective strategy that can mirror the TYC strategy. In theory linear harvesting (and harvesting at various density dependent rates), seem to work better than a mimic of TYC where males are stocked and females harvested (FHMS) \citep{L18,L19}. However, the impact of Allee effects on the overall success or failure of such a class of strategies remains unexplored.

Allee effects are positive relationships between individual fitness and population density. These can be strong, where there is a threshold, below which the population growth rate is negative. They could also be weak, where the growth rate is always positive, but smaller at lower densities \citep{Cet99,SS99}. In the context of marine fishes, researchers observe that an Allee effect is significant at very low population size and with bias in sex ratio \citep{PK17,W12}. Researchers \citep{Net18,PK17} observed that extinction of Atlantic cod ({\it Gadus morhua}) in the southern Gulf of St. Lawrence and the depletion of Atlantic herring ({\it Clupea harengus}) population in the North Sea are due to predation-driven Allee effect. For a population with male-biased sex ratio would lead to difficulty in finding a mate, even for species that use powerful sex pheromones. Such skewed sex ratios fortify Allee effects on account of mating failure, prompting the risk of populace extinction. In sex structured (into male and female) population models, specially in fishes, having a weak Allee effect only on the female is quite feasible, as at low female densities, we would expect smaller clutch sizes - however a few males could fertilize a large number of females - so clutch sizes could still be large \citep{Aet04}.

The dynamics of sex structured two species mating models, even with the inclusion of Allee effects, is generically like Fig. \ref{fig01} (a). There are typically two interior equilibria, one unstable (saddle) and one locally stable, also the extinction equilibrium is locally stable. The stable manifold of the saddle (separatrix) splits the phase space into two sections, delineated by the extinction boundary, also called the allee threshold or threshold manifold in the literature \citep{B02,SJ09}. If one picks initial data on one side of this curve, solutions tend to the stable interior equilibrium, and if one picks initial data on the other side of the boundary, solutions tend to the extinction equilibrium. Note, although the curve is seen to be of \emph{hyperbolic} shape (that is monotone with respect to initial conditions, in the phase space), the general shape of this curve, even in two species mating models, with or without Allee effects, is an open problem in ecology.

Also, by considering Allee effect in invasive fish population and a continued harvesting/stocking, the rarity of wild-type females would lead to difficulty in finding mates, and so the invasive fish population would eventually become locally extinct. However, it is not economically viable to harvest/stock indefinitely. Thus determining the time for terminating harvesting/stocking is critical as the wild-type invasive species would either go extinct or recover, such as in the absence of supermale invasive fish \citep{WP14}.

In the current manuscript we show that,
\begin{itemize}
	\item Both a saddle-node and Homoclinic bifurcation can occur in the FHMS model with weak Allee effect via Lemma \ref{lem:sn1} and see Fig. \ref{AfterHomoclinic}. We show limit cycle dynamics is not possible without the weak Allee effect in place, via Lemma \ref{lem:lc1}, however the weak Allee effect can lead to limit cycle dynamics, via Lemma \ref{lem:lc12}.

	\item The FHMS model with a weak Allee effect, can exhibit an extinction boundary (Allee threshold) that is non-hyperbolic. Such dynamics can enable extinction, essentially for any initial data. This is completely different when a weak Allee effect is \emph{not} in place. See Fig. \ref{fig:nwa}. A non-hyperbolic extinction boundary is also possible via a strong Allee effect, see Fig. \ref{fig:Bendings}, but the ``bending" of the boundary is not as pronounced as in the weak Allee effect case.

	\item We show when/if harvesting/stocking can be terminated at certain finite time, and when the population of invasive fish is below some threshold, to yield extinction. See Figs. \ref{fig04} - \ref{fig05}.

	\item We consider the spatially explicit FHMS model with weak Allee effect. We show that the weak Allee effect can cause Turing instability, and impossibility without it, via Lemma \ref{lem:nt1}, and Theorem \ref{thm:t1}.

	\item We discuss the implications of our results to biocontrol, via these strategies.
\end{itemize}

\section {Background}
Here we recap the basic TYC and FHMS models as presented in \citep{L18} and \citep{Parshad10}.
\subsection{The TYC Model}
In the TYC strategy, supermales ($S$) of the invasive species with two $YY$ Chromosomes, are introduced into the wild-type invasive fish population having wild-type males ($M$) and females ($F$). The rate of injection of supermale invasive fish is taken as $\mu_0$ (population time$^{-1}$). The reproduction rate of wild-type invasive fish species due to the interactions between male and female wild-type invasive fish species is $\beta_1$ (population$^{-1}$ time$^{-1}$), whereas the reproduction rate of wild-type invasive fish due to the interactions between wild-type female invasive fish and supermale invasive fish is $\beta_2$ (population$^{-1}$ time$^{-1}$). The death rates of wild-type and the supermale invasive fish are taken as $\delta_1$ (time$^{-1}$) and $\delta_2$ (time$^{-1}$) respectively. The carrying capacity of the system is $K_1$ (population), and the logistic term $L=1-\frac{F+M+S}{K_1}$ is used to constrain the invasive fish population. The equations describing the TYC model are:
\begin{eqnarray} {\label{eq:1}}\nonumber
  \frac{dF}{dT} &=& \frac{1}{2}\beta_1 F M L-\delta_1 F \nonumber\\
  \frac{dM}{dT} &=& F \left(\frac{1}{2}\beta_1M+\beta_2 S\right)L-\delta_1 M\\
  \frac{dS}{dT} &=& \mu_{0}-\delta_2 S, \nonumber
\end{eqnarray}
where $F(0)\geq 0$, $M(0)\geq 0$ and $S(0)\geq 0$.
\subsection{Existence and stability of equilibria when \texorpdfstring{$\mu_0=0$}{Lg}}
We refer the reader to \citep{WP14} for detailed analysis on the existence and stability of equilibria to system (\ref{eq:1}).
The equilibria to  system (\ref{eq:1}) after nondimensionalization are  $E_0=(0,0)$ and $E_{1,2}=(f^*_{\pm}, m ^*_{\pm})$  with $f^*=m^*$ where
$$ f^*_{\pm}  = \dfrac{1}{4} \Big( 1 \pm \sqrt{1-\Phi} \Big), \text{where} ~~~ \Phi=\dfrac{8}{\rho}.   $$

We recap some standard results on the model \citep{WP14},

\begin{theorem}.
If $\Phi <1$,
\begin{itemize}
\item[(i)] the extinction state $E_0$ is locally stable.\\
\item[(ii)] the equilibrium point $E_1$ is locally stable.\\
\item[(iii)] the equilibrium point $E_2$ is locally unstable.
\end{itemize}
\end{theorem}

\begin{remark}.
When $\rho=8,$  the two interior equilibrium points $E_1$ and $E_2$ collide with each other giving rise to a saddle-node bifurcation.
\end{remark}

\subsection {The FHMS Model}

A key issue in the implementation of the TYC strategy is the production of supermales. Estrogen-induced feminization of the wild-type male fish has often proven inefficient to obtain sex-reversed $XY$ physiological females. In such a situation, we can implement a sex-skewing strategy by removing a fraction of wild-type females by means of harvesting whilst adding in the wild-type males by means of stocking. This strategy is called female harvesting male stocking (FHMS), first proposed by Lyu \citep{L18,L19}.

For the FHMS model described below, the primary sex ratio in offspring is denoted by $r$ $(0<r<1)$. The harvesting rate of females and the stocking rate of the males are denoted by $h_F$ and $s_M$ respectively. Also, we assume that $0\le s_M<\delta$. The equations describing the FHMS system are:

\begin{eqnarray} {\label{eq:2}}\nonumber
  \frac{dF}{dT} &=& r\beta F M L-(\delta + h_F)F \nonumber\\
  \frac{dM}{dT} &=& (1-r)\beta F M L + (s_M-\delta) M,
\end{eqnarray}
where $F(0)\geq 0$ and $M(0)\geq 0$.

In order to reduce the number of parameters, we introduce dimensionless variables $$f=\frac{F}{K_1},\;m=\frac{M}{K_1},\; t=T\delta ,$$ and the dimensionless parameters $$\alpha=\frac{\beta K_1}{\delta},\; h=\frac{h_F}{\delta},\; s=\frac{s_M}{\delta}.$$
With these substitutions, the equations describing the system become:
\begin{eqnarray} {\label{eq:3}}\nonumber
  \frac{df}{dt} &=& r\alpha fm \left(1-f-m\right)- (1+h)f\equiv F_1(f,m) \nonumber\\
  \frac{dm}{dt} &=& (1-r)\alpha fm \left(1-f-m\right)+(s-1) m\equiv F_2(f,m),
\end{eqnarray}
where $f(0)\geq 0$ and $m(0)\geq 0$.

We recap the following results \citep{L18,L19},

\begin{lemma}.
If $f(0)$ and $m(0)$ are positive, then all possible solutions of the system \eqref{eq:3} are non-negative.
\end{lemma}

\begin{lemma}.
All the solutions of the system \eqref{eq:3} are contained in some bounded subset in the plane $$\left\{(f,m)\in \textbf{R}^2: f\geq 0,\;\; m\geq 0\right\}.$$
\end{lemma}

\subsection {Equilibria and their stability}

The equilibria and stability analysis of system \eqref{eq:3} are presented in Appendix A.
We state some results on \eqref{eq:3}  that were not shown in \citep{L19},

\begin{lemma}.
For $h_*<h<h^*$, the invasive species get eliminated from the system \eqref{eq:3} via saddle-node bifurcation when $s$ is decreased through $s=s^*$.
\end{lemma}

\begin{proof}.
At $s=s^*$, we have $f^*=\frac{1}{2(1+\mu)}$ where $\mu=\frac{(1-r)(1+h)}{r(1-s)}$  and so $\mbox{\mbox{Det}}(J^*){_{|_{s=s^*}}}=0$, where $	J^*$ is the Jacobian of system \eqref{eq:3}. Since $\mbox{Tr}(J^*_i)<0$, it follows that $J^*_i$ has a simple eigenvalue at $s=s^*$.\\
Let $F(f,m;s)=\left(F_1\;\; F_2\right)^T$ and $V$ and $W$ are the eigenvectors corresponding to the zero eigenvalue for $J^*{_|{_{s=s^*}}}$ and $J^{*T}{_|{_{s=s^*}}}$ respectively.  Then we have $F_s(f,m;s)=\left(0\;\; \frac{\mu}{2(1+\mu)}\right)^T$, $U=\left(1 \;\; \mu\right)^T$ and $V=\left(1\;\; \frac{r}{\mu(1-r)}\right)^T$.
This gives $V^T F_{s}\left(E^*_i;s^*\right)=\frac{r}{2(1+\mu)(1-r)}\neq 0$ and $V^T\left[D^2F\left(E^*_i;s^*\right)(U,U)\right]=-r\alpha(1+\mu)\neq 0$. Therefore, by Sotomayor's theorem \citep{P13} it follows that the system \eqref{eq:3} undergoes a saddle-node bifurcation at $E^*_i$ when $s$ crosses $s^*$.
\end{proof}

\begin{lemma}.
\label{lem:lc1}
System (\ref{eq:3}) has no limit cycles.
\end{lemma}

\begin{proof}.
We use the Dulac theorem to show the non-existence of limit cycles in (\ref{eq:3}). Let us consider the Dulac function
\begin{eqnarray} {\label{eq:6}}
\Psi(f,m)=\dfrac{1}{fm}
\end{eqnarray}
where $f\neq0 $ and $m \neq 0$.
Then,
\begin{eqnarray}\nonumber
\dfrac{\partial ( F^1 \Psi)}{\partial f} +\dfrac{\partial  (F^2 \Psi)}{\partial m}
 &=& \dfrac{\partial }{\partial f}\Big(r \alpha (1-f-m)-\dfrac{1+h}{m}  \Big) + \dfrac{\partial }{\partial m}\Big((1-r)\alpha (1-f-m) +\dfrac{s-1}{f} \Big) \nonumber \nonumber\\
&=& -r \alpha - (1-r)\alpha  \nonumber \\
&=& -\alpha < 0. \nonumber
\end{eqnarray}

Hence the system (\ref{eq:3}) has no limit cycles. This completes the proof.
\end{proof}

\begin{lemma}.
If $h=0=s$, then system (\ref{eq:3}) has no limit cycles.
\end{lemma}

\begin{proof}.
Suppose $h=0=s$ in  system (\ref{eq:3}). We use the Dulac theorem again and  consider the Dulac function in Eq.(\ref{eq:6}). Then,
\begin{equation}\nonumber
\dfrac{\partial ( F^1 \Psi)}{\partial f} +\dfrac{\partial  (F^2 \Psi)}{\partial m} \\
= \dfrac{\partial }{\partial f}\Big(r \alpha(1-f-m)-\dfrac{1}{m}  \Big) + \dfrac{\partial }{\partial m}\Big((1-r)\alpha (1-f-m) -\dfrac{1}{f} \Big) \\
=  -r \alpha - (1-r)\alpha = -\alpha\\
< 0.
\end{equation}

Hence the system (\ref{eq:3}) has no limit cycles for $h=0=s$. This completes the proof.
\end{proof}

\section{The FHMS Model with a weak Allee Effect}

We modify the system \eqref{eq:3} by considering Allee effect in the population. The equations describing the FHMS system with a weak Allee effect are given by:
\begin{eqnarray} {\label{eq:4}}\nonumber
  \frac{dF}{dT} &=& r\beta F^2 M L-(\delta + h_F)F \nonumber\\
  \frac{dM}{dT} &=& (1-r)\beta F^2 M L + (s_M-\delta) M,
\end{eqnarray}
where $F(0)\geq 0$ and $M(0)\geq 0$.\\

We introduce dimensionless variables $$f=\frac{F}{K_1},\;m=\frac{M}{K_1},\; t=T\delta ,$$ and the dimensionless parameters $$\alpha=\frac{\beta K^2_1}{\delta},\; h=\frac{h_F}{\delta},\; s=\frac{s_M}{\delta}.$$
With these substitutions, the equations describing the system become:
\begin{eqnarray} {\label{eq:5}}\nonumber
  \frac{df}{dt} &=& r\alpha f^2m \left(1-f-m\right)- (1+h)f\equiv \bar{F}_1(f,m) \nonumber\\
  \frac{dm}{dt} &=& (1-r)\alpha f^2m \left(1-f-m\right)+(s-1) m\equiv \bar{F}_2(f,m),
\end{eqnarray}
where $f(0)\geq 0$ and $m(0)\geq 0$.

\begin{figure}
\centering
\hspace*{-1cm}
\begin{tabular}{cc}
\includegraphics[scale=0.25]{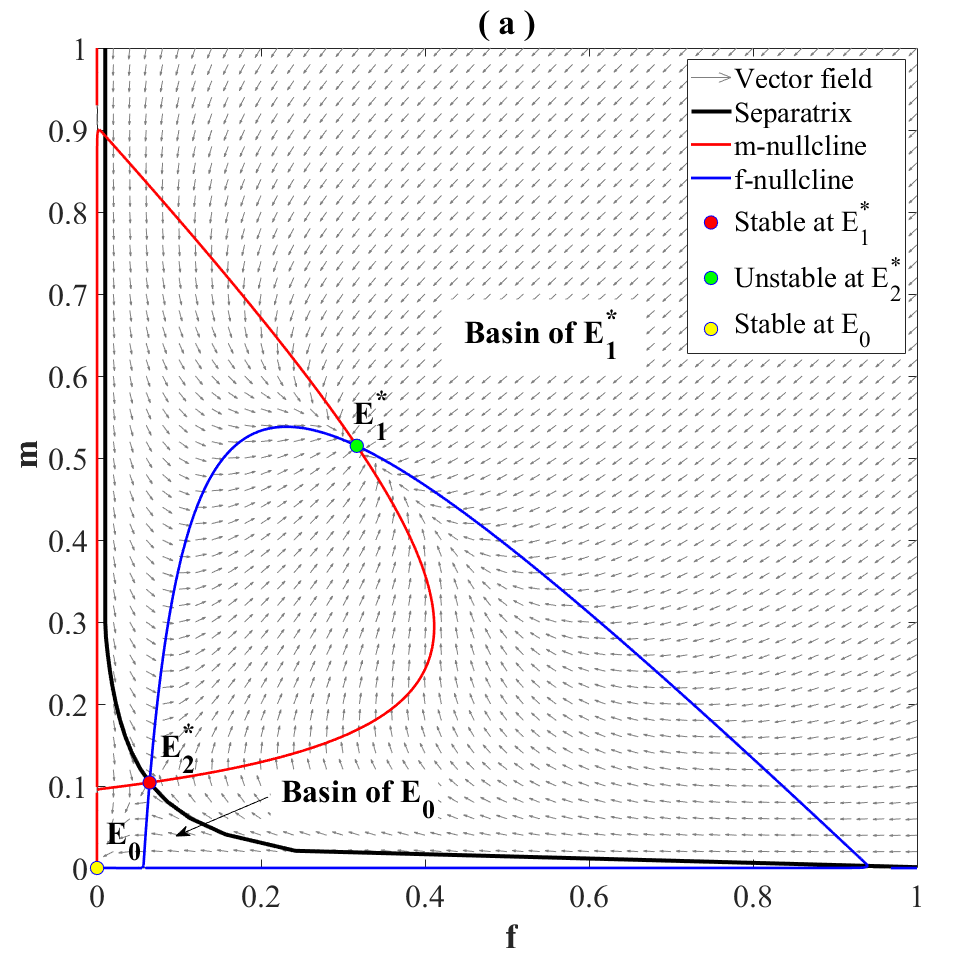} & \includegraphics[scale=0.25]{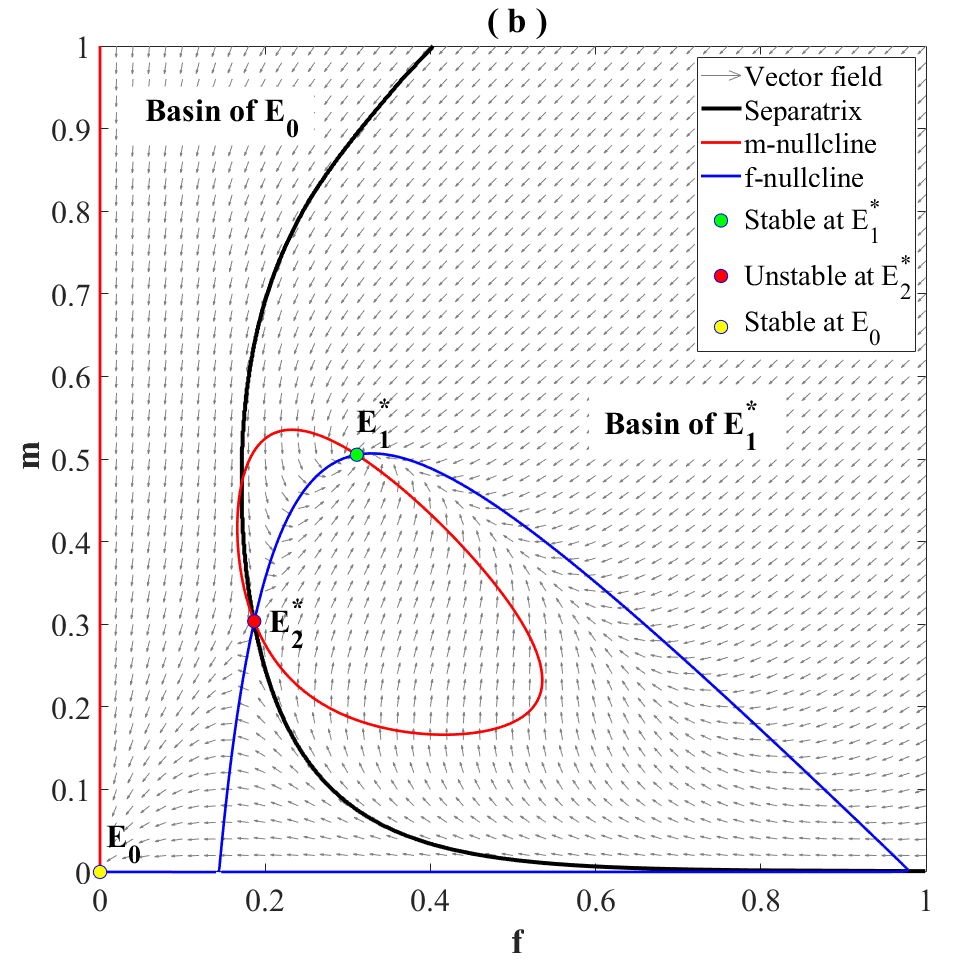}
\end{tabular}
\caption{The separatrix divides the phase plane into the regions of extinction and recovery for the $(a)$ FHMS system \eqref{eq:3} without Allee and $(b)$ FHMS system \eqref{eq:5} with weak Allee} \label{fig01}
\end{figure}

\begin{lemma}.
If $f(0)$ and $m(0)$ are positive, then all possible solutions of the system \eqref{eq:5} are non-negative.
\end{lemma}
\begin{proof}.
 We have $f(t)=f(0)e^{\displaystyle \int_0^t\phi_1(f,m;\tau)d\tau}$ and $m(t)> m(0)e^{\displaystyle \int_0^t\phi_2(f,m;\tau)d\tau}$, where \\
$r \left\{m\phi_2(f,m;t)+1-s\right\}=(1-r) \left\{f\phi_1(f,m;t)+1+h\right\}$. Here $\phi_1= r\alpha fm \left(1-f-m\right)- (1+h)$ and $\phi_2=(1-r)\alpha f^2 \left(1-f-m\right)+(s-1). $\\
This implies, all solutions of \eqref{eq:5} remain within $\left\{(f,m)\in \textbf{R}^2: f\geq 0, m\geq 0\right\}$ starting from an interior point of it. Therefore, $\textbf{R}^2_+=\left\{(f,m)\in \textbf{R}^2: f> 0, m> 0\right\}$ is an invariant region, and as long as $f(t)>0$ and $m(t)>0$ for all $t$, the local existence and uniqueness properties hold in $\textbf{R}^2_+$. We now prove that the solutions of \eqref{eq:5} with initial values in $\textbf{R}^2_+$ are bounded, so that the system \eqref{eq:5} is biologically meaningful.
\end{proof}
\begin{lemma}.
All the solutions of the system \eqref{eq:5} are contained in some bounded subset in the plane $$\left\{(f,m)\in \textbf{R}^2: f\geq 0,\;\; m\geq 0\right\}.$$
\end{lemma}

The proof is shown in section \ref{AppB}.

\subsection {Equilibria and their stability}
System \eqref{eq:5} has the nullclines $\bar{F}_i=0$ $(i=1,2)$. Solving these nullclines yields the following equilibria:
\begin{itemize}
\item[(i)] Invasive fish-free equilibrium $E_0=\left(0,0\right)$ exists always and is locally asymptotically stable (as $s<1$).

\item[(ii)] coexistence equilibria $E^*_{i}=\left(f^*_i,\mu f^*_i\right)$, where $\mu=\frac{(1-r)(1+h)}{r(1-s)}$ and $f^*_i$ is a positive root of the equation
$$G(f)\equiv f^3-\frac{1}{1+\mu}f^2+\frac{1+h}{r\alpha\mu(1+\mu)}=0.$$
\end{itemize}

The stability of system \eqref{eq:5} is determined by using eigenvalue analysis of the Jacobian matrix evaluated at the appropriate equilibrium. The eigenvalues of the Jacobian matrix $(J(E_0))$ of the system \eqref{eq:5} at $E_0$ are $s-1$ and $-(1+h)$. Since $0\le s<1$, all the eigenvalues of $J(E_0)$ are negative. This gives the following lemma:
\begin{lemma}.
The invasive fish-free equilibrium $E_0$ is always locally asymptotically stable.
\end{lemma}
The Jacobian of the system \eqref{eq:5} evaluated at $E^*_i$ is given by
$$J^*_i=\begin{pmatrix}
  r\alpha \mu f^{*2}_i\left(1-2f^*_i-\mu f^*_i\right) & r\alpha f^*_i(1-f^*_i-2\mu f^*_i)\\
  (1-r)\alpha \mu f^{*2}_i(2-3f^*_i-2\mu f^*_i) & -(1-r)\alpha \mu f^{*3}_i
\end{pmatrix}.$$
The system \eqref{eq:5} is locally asymptotically stable at $E^*_i$ if and only if $\mbox{Tr}(J^*)<0$ and $\mbox{Det}(J^*_i)>0$. \\

From Fig. \ref{fig03} it is observed that as the harvesting rate ($h$) is increased, the female-male gender ratio drops significantly. Once $h$ crosses some critical threshold value, there leads to a sudden change of transition from stable coexistence state to invasive fish-free state (cf. Fig. \ref{fig02}). Therefore, it is necessary to study the behaviour of the system \eqref{eq:5} by considering $h$ as a bifurcation parameter. \\

\begin{figure}
\centering
\hspace*{-1.5cm}
\begin{tabular}{cc}
 \includegraphics[scale=0.25]{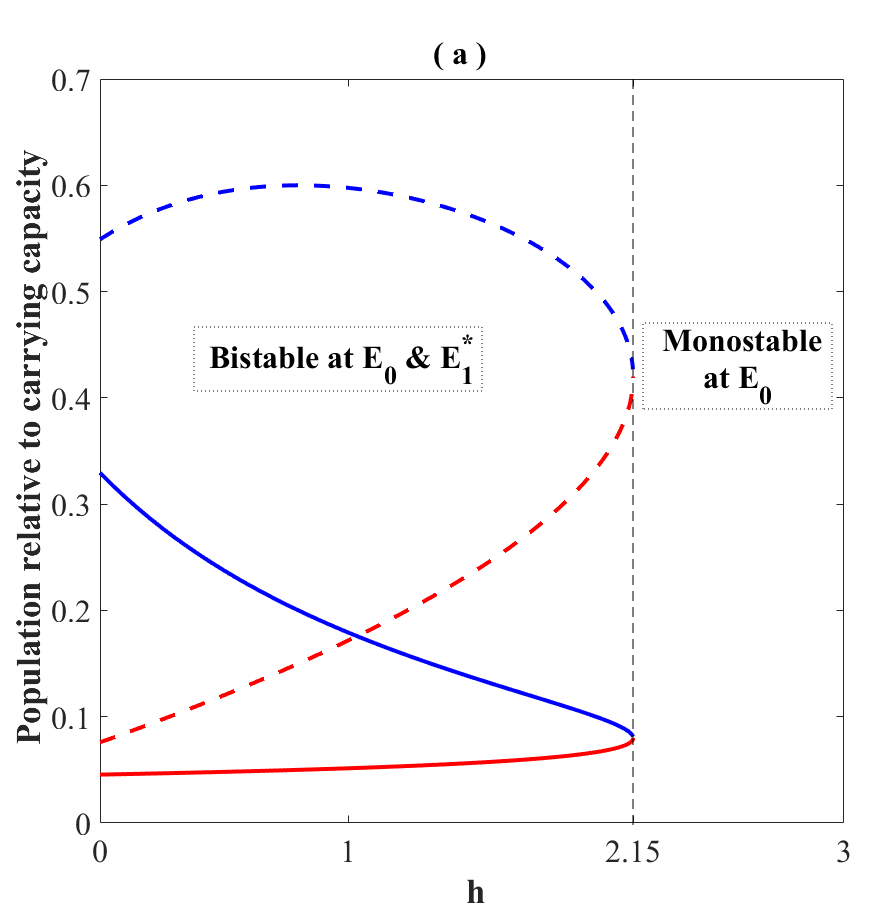}& 
 \includegraphics[scale=0.25]{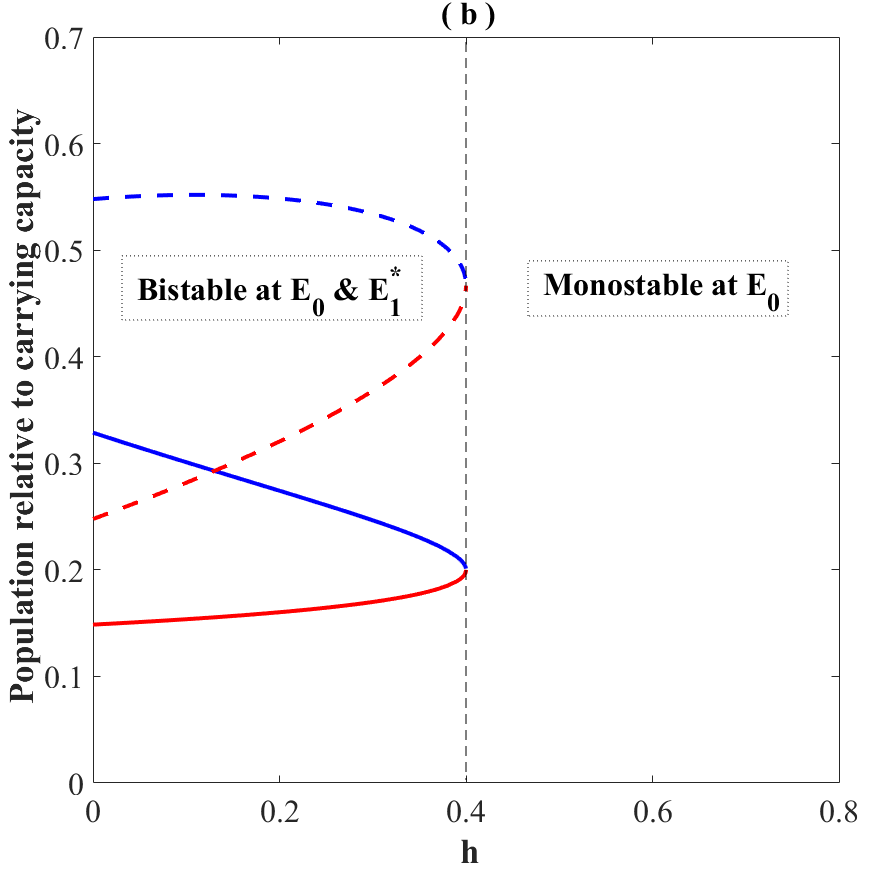}
 \end{tabular}
\caption{Bifurcation diagrams of the $(a)$ FHMS model without Allee effect \eqref{eq:3} and $(a)$ FHMS model with weak Allee effect \eqref{eq:5}, where $h$ as the bifurcation parameter (XX and XY are in solid and dotted curves respectively). The stable and unstable equilibrium are represented by the blue curve and the red curve, respectively }\label{fig02}
\end{figure}

\begin{lemma}.
\label{lem:sn1}
The invasive species get eliminated from the system \eqref{eq:5} via saddle-node bifurcation when $h$ is increased through $h=h^*$.
\end{lemma}

The proof is given in the appendix section \ref{AppB1}.

\begin{lemma}.
\label{lem:lc12}
Consider the  Jacobian of system (\ref{eq:5}). If the following hold
 \begin{itemize}
	\item[(i)] $\mbox{Tr}(J_i^*)\big |_{E_1^*}=0$,\\
        \item[(ii)]$\text{\mbox{Det}}(J_i^*)\big |_{E_1^*}  >0,$ \\
 \item[(iii)] $\dfrac{d}{ds}(\text{\mbox{Tr}}(J_i^*))\big |_{E_1^*} \neq 0  $ at $s=s_{\text{hf}}$,
 \end{itemize}

then the system \eqref{eq:5} exhibits periodic oscillation via a Hopf bifurcation when $s$ is increased through $s=s^*$.
\end{lemma}

The proof is given in the appendix section \ref{AppB3}.

\begin{figure}
\centering
\hspace*{-1.5cm}
\begin{tabular}{cc}
\includegraphics[scale=0.28]{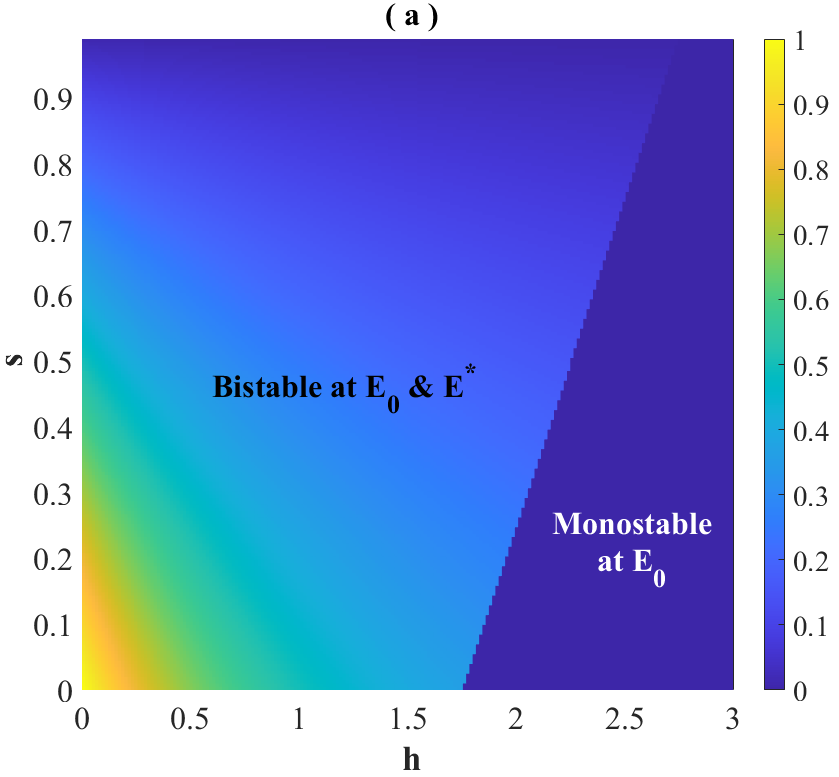}&
\includegraphics[scale=0.28]{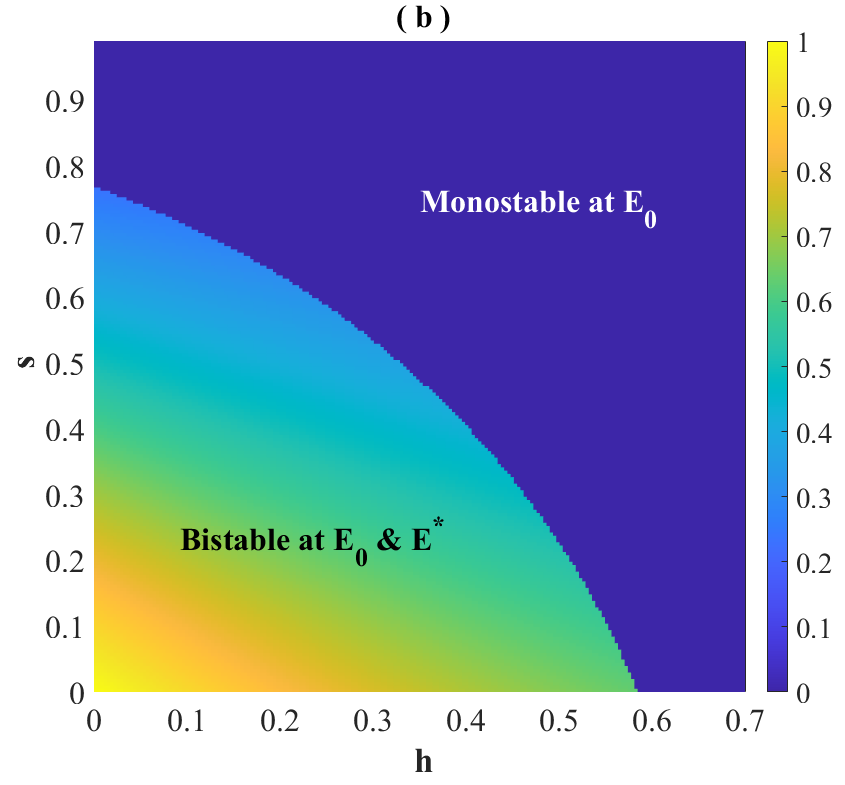}
\end{tabular}
\caption{The changes in the female-male gender ratio with the changes in the scaled harvesting rate ($h$) and the scaled stocking rate ($s$) for the $(a)$ FHMS model without Allee effect \eqref{eq:3} and $(b)$ FHMS model with weak Allee effect \eqref{eq:5}}\label{fig03}
\end{figure}

\subsection{Intermediate Harvesting}

As a result of the Allee effect in the invasive fish population, with the harvesting rate $h\geq h^*$, the eradication of invasive fish can be successfully achieved by stopping the harvesting and stocking once the female invasive fish population fall below some threshold value. The minimum number of females needed for its survival would give the minimum time of continuous harvesting and stocking in order to eradicate invasive fish species. Once the harvesting and stocking of the invasive fish are stopped, the dynamics of the system will be governed by system \eqref{eq:5} with $h=0$ and $s=0$. Each of these systems has a stable fish-free boundary equilibrium. In the phase plane, once the solution trajectories enter the basin of attraction of the boundary equilibrium, the fish species become extinct. Determining the threshold of the female fish population is critical as after discontinuing the harvesting and stocking, the invasive fish species would either go extinct or recover.

It has been observed that with the female-harvesting rate at $h^*$ and with different initial sex ratios, the minimum threshold female population for survivability lies in between $0.1$ and $0.15$ (cf. Fig. \ref{fig04}). It is seen that at a given supply rate, the minimum harvesting time of female invasive fish varies with the initial population densities of invasive fish. While the minimum time for the continuous harvesting of females increases with the initial population of equal sex ratio (cf. Figs. \ref{fig04}$(a)\;\&\;\ref{fig04}(b)$), population densities initially sex-skewed towards males require less harvesting time than the population sex-skewed towards females (cf. Figs. \ref{fig04}$(c)\;\&\;\ref{fig04}(d)$). For $h=0=s$ and $h=h^*$, let the separatrices of extinction and recovery of the invasive fish population be in the form $f=\Gamma^0_{h}(m)$ and $f=\Gamma^*_h(m)$ respectively. Then the regions of extinction for $h=0=s$ and $h=h^*$ in the $fm$-plane are given by: \\
$R^0_{h}=\left\{(f,\;m)\in R^2: 0\le f\le \Gamma^0_{hs}(m), \; 0\le m\le 1\right\}$ and \\
$R^*_h=\left\{(f,\;m)\in R^2: 0\le f\le \Gamma^*_h(m), \; 0\le m\le 1\right\}$ respectively.\\
The eradication strategy would be to discontinue harvesting and stocking of invasive fish once the invasive fish population enters the region $R^0_{h}\bigcap R^*_h$ (cf. Fig. \ref{fig05}).

At a given initial population density of the invasive fish species and the harvesting rate of female invasive fish at its minimum admissible value $h^*$, this technique will give the least possible continuous harvesting or (and) stocking time in order to eliminate the invasive fish species from the system. Fig. \ref{fig05} gives the minimum continuous harvesting or (and) stocking time of the invasive fish for the complete removal of the invasive fish population with different initial population densities.

\begin{figure}
\centering
\hspace*{-1cm}
\begin{tabular}{cc}
\includegraphics[scale=0.2]{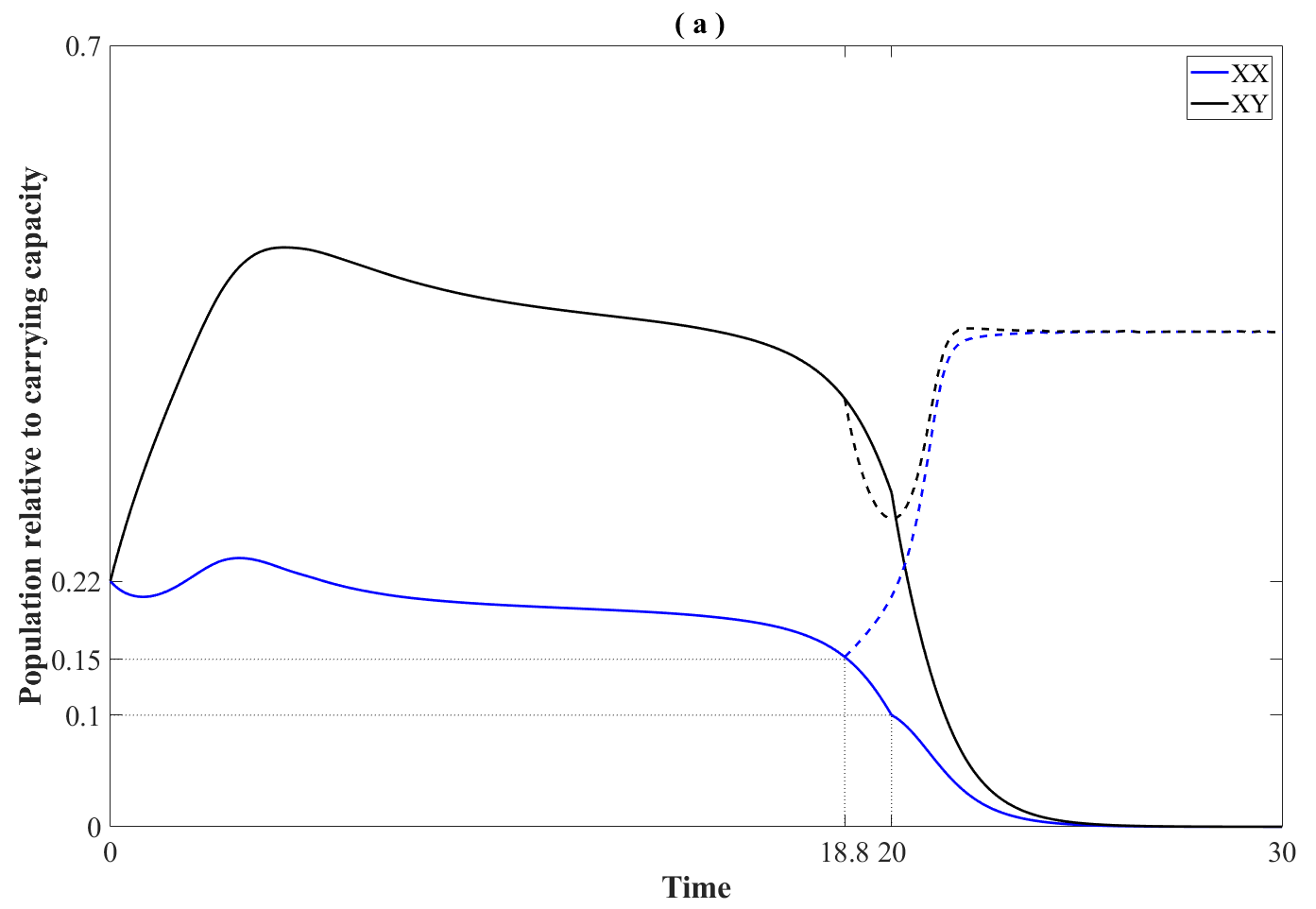} &
\includegraphics[scale=0.2]{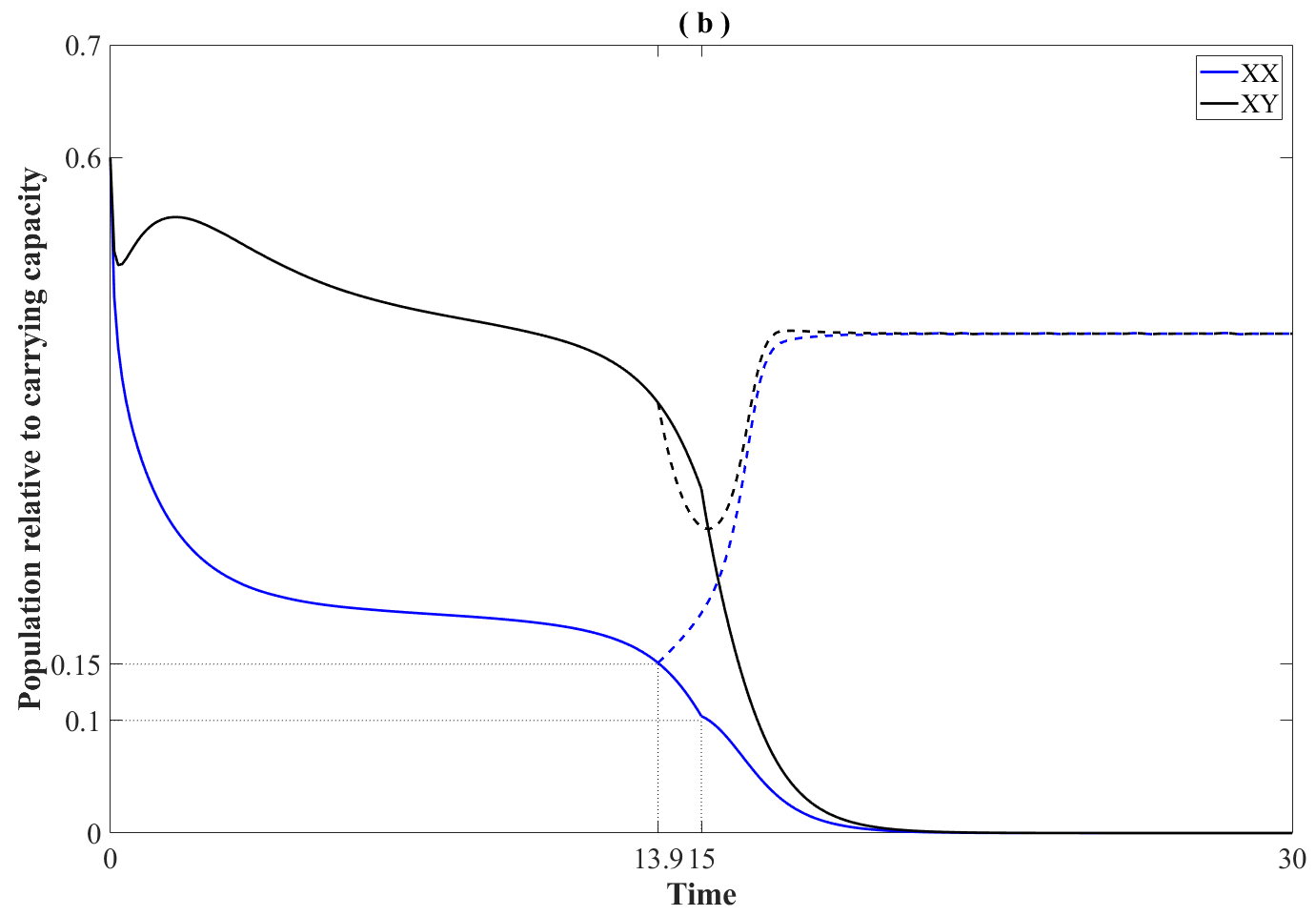}\\
\end{tabular}
\hspace*{-1cm}
\begin{tabular}{cc}
\includegraphics[scale=0.2]{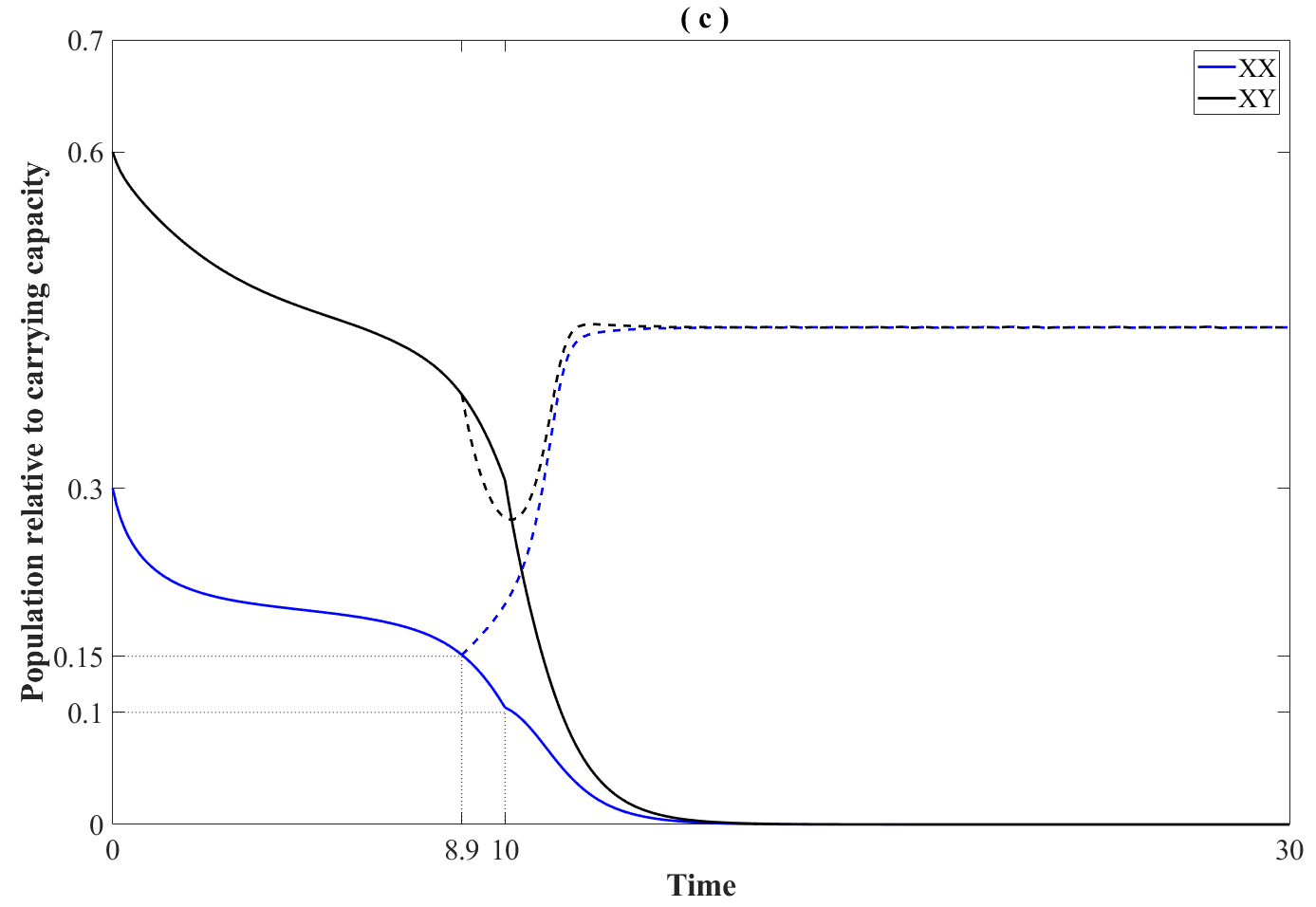}&
 \includegraphics[scale=0.2]{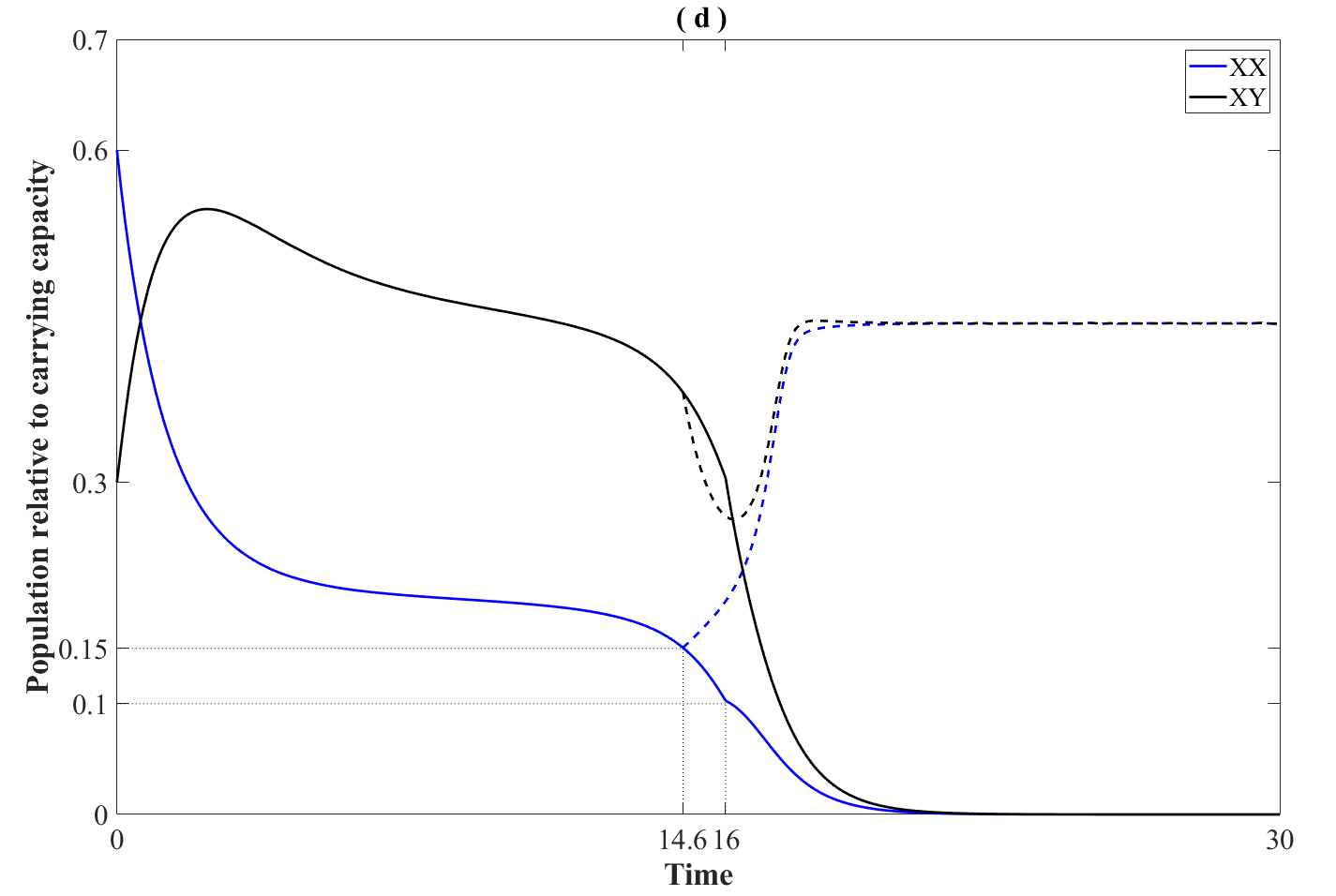}\\
\end{tabular}
\caption{The time for extinction and recovery of invasive fish population with $(a)$ low and  $(b)$ high initial population densities at $1 : 1$ sex-ratio. Time for extinction and recovery of invasive fish population with initial wild-type population sex-ratio is $(c)$ skewed towards males $(f(0):m(0)=1:2)$ and $(d)$ skewed towards females $(f(0):m(0)=2:1)$. The harvesting rate of the female species is kept at its minimum admissible value $h^*$ for eradication. The plots for extinction and recovery are represented by solid and dashed lines, respectively}\label{fig04}
\end{figure}

\begin{figure}
\centering
\hspace*{-1cm}
\begin{tabular}{cc}
\includegraphics[scale=0.32]{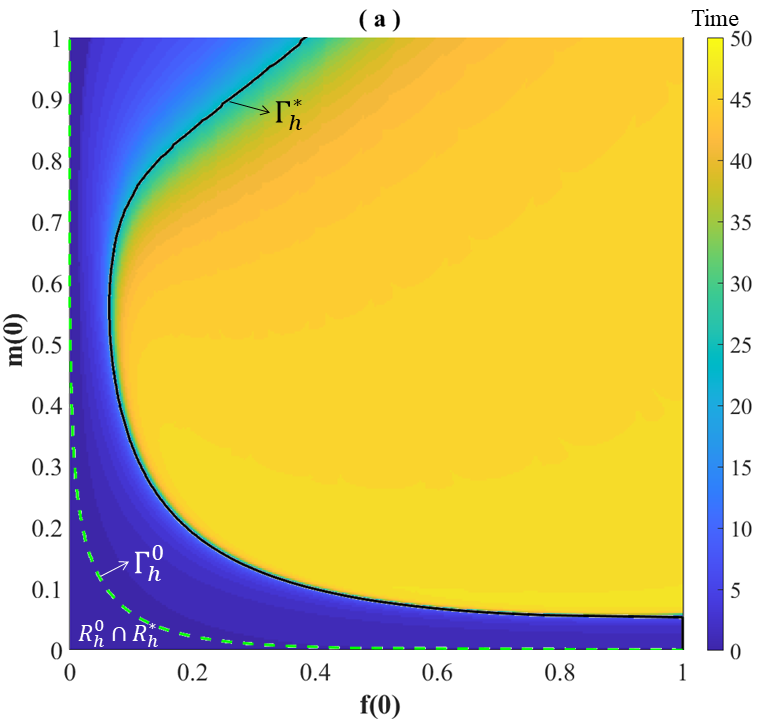}&
\includegraphics[scale=0.32]{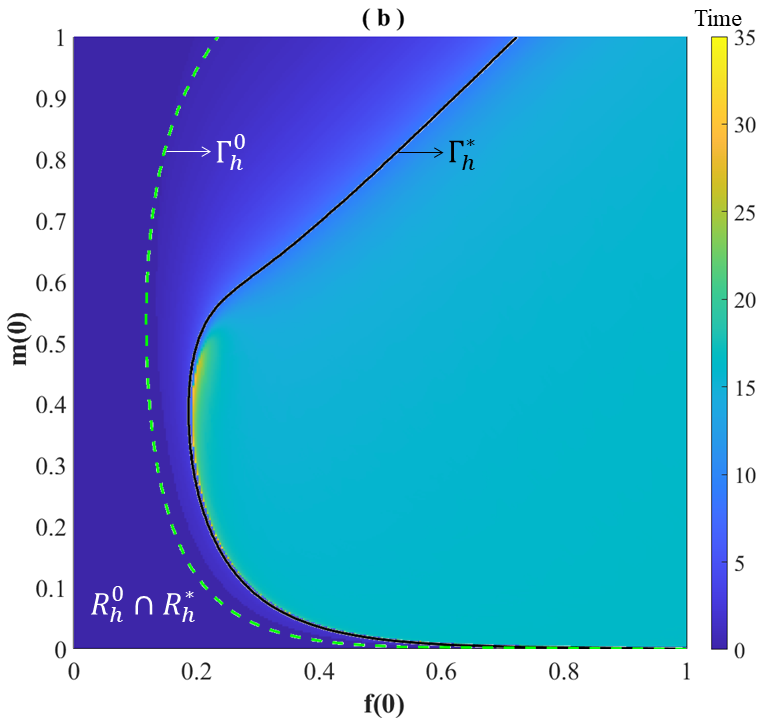}\\
\end{tabular}
\caption{The minimum continuous harvesting and stocking time for the elimination of the invasive fish population with different initial population densities for the $(a)$ FHMS system \eqref{eq:3} without Allee and $(b)$ FHMS system \eqref{eq:5} with weak Allee. The separatrices for the systems with $h=h^*$ and $h=0=s$ are denoted by the curves $\Gamma^*_h$ and $\Gamma^0_{h}$ respectively}
\label{fig05}
\end{figure}

\newpage

\subsection{The Extinction Boundary}

As mentioned earlier, the dynamics of sex structured two species mating models is generically like Fig. \ref{fig01} (a). The stable manifold of the saddle (separatrix) divides the phase space, delineated by an extinction boundary, also called ``Allee" threshold (when there are Allee effects included) \citep{B02} or threshold manifold \citep{SJ09}. If one picks initial data on one side of the boundary, solutions tend to the stable interior equilibrium, and if on the other side, solutions tend to the extinction equilibrium. Note, it is always seen to be of \emph{hyperbolic} shape (monotone w.r.t initial data). The literature is rife with examples of such mating models - where the extinction boundary always turns out to be hyperbolic \citep{B08,B02,B04}. Notably, Schreiber rigorously proved, that there exists a hyperbolic extinction boundary in the case that certain sufficient conditions on the system concerned are met \citep{S04} - strong monotonicity of the system being one of them. The results have been proved for general monotone systems as well \citep{ST01,SH06}. In the event that the sufficient conditions of \citep{S04} are not met, the shape of the Allee threshold remains an open problem \citep{B04}. In particular, in cases where a strong Allee effect has been introduced into TYC type mating dynamics, we still see a hyperbolic extinction boundary \citep{BPL20,JB20}, but this might be due to a lesser exploration of the parameter space.

\begin{remark}.
Consider a population system $\dot{x}= xG(x)$, where $x=(f,m)$. Then \citep{S04}, requires that if $\exp(G(0,0))$ is primitive (and some other technical conditions are met), then that implies that there exists a hyperbolic extinction boundary. In our case,
\vspace{-0.15cm}
\begin{equation}
\exp(G(0,0)) =
\begin{pmatrix}
\exp(-(1+h)) &  0\\
0 & \exp(s-1)
\end{pmatrix}.
\end{equation}
This is not primitive. Hence the shape of the exact extinction boundary in our case is unknown.
\end{remark}

We show, via numerical simulation that the FHMS model with weak Allee effect, can exhibit interesting dynamics in that the extinction boundary may not be hyperbolic. In Fig. \ref{fig:nwa}, we increase the stocking parameter $s$ and observe that the extinction boundary changes from the standard hyperbolic shape, to a loop. In Fig. \ref{AfterHomoclinic} we see a limit cycle form, collide with the saddle, the stable manifold, the unstable manifold and form a homoclinic orbit - via a homoclinic bifurcation.

\begin{figure}
   \centering
  \hspace*{-1.cm}
\begin{tabular}{ccc}
\includegraphics[scale=0.13]{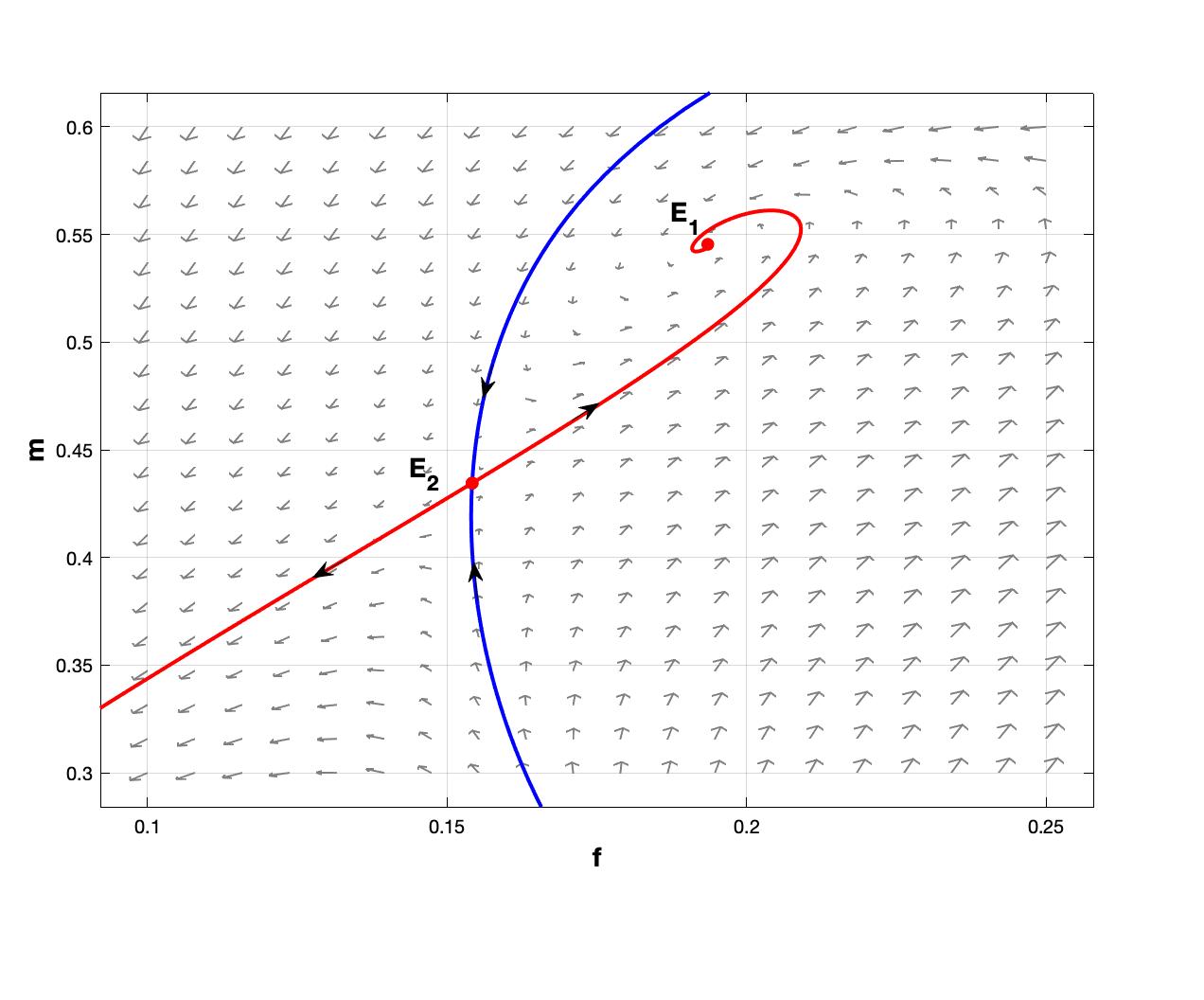} &
\includegraphics[scale=0.13]{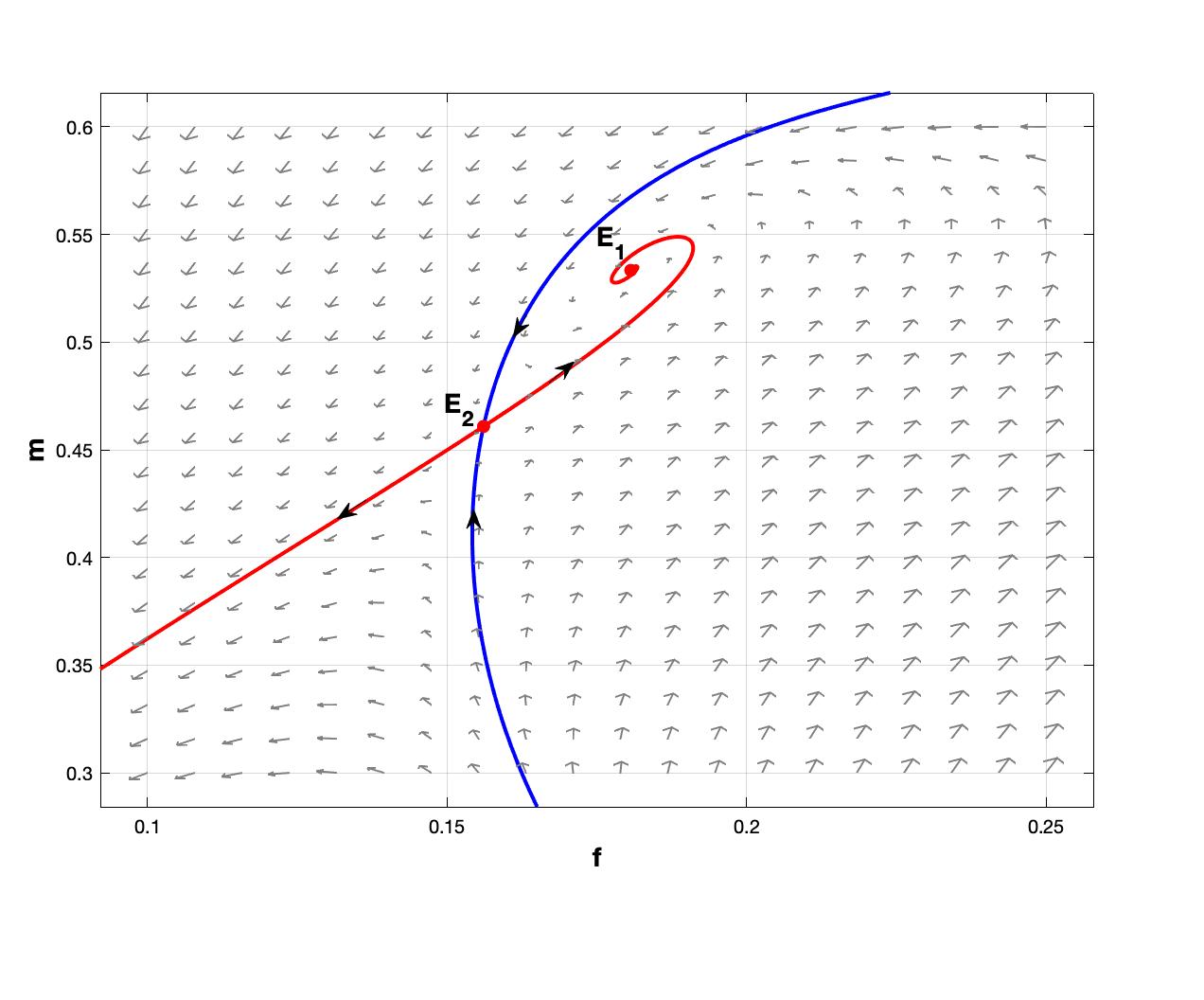} &
\includegraphics[scale=0.13]{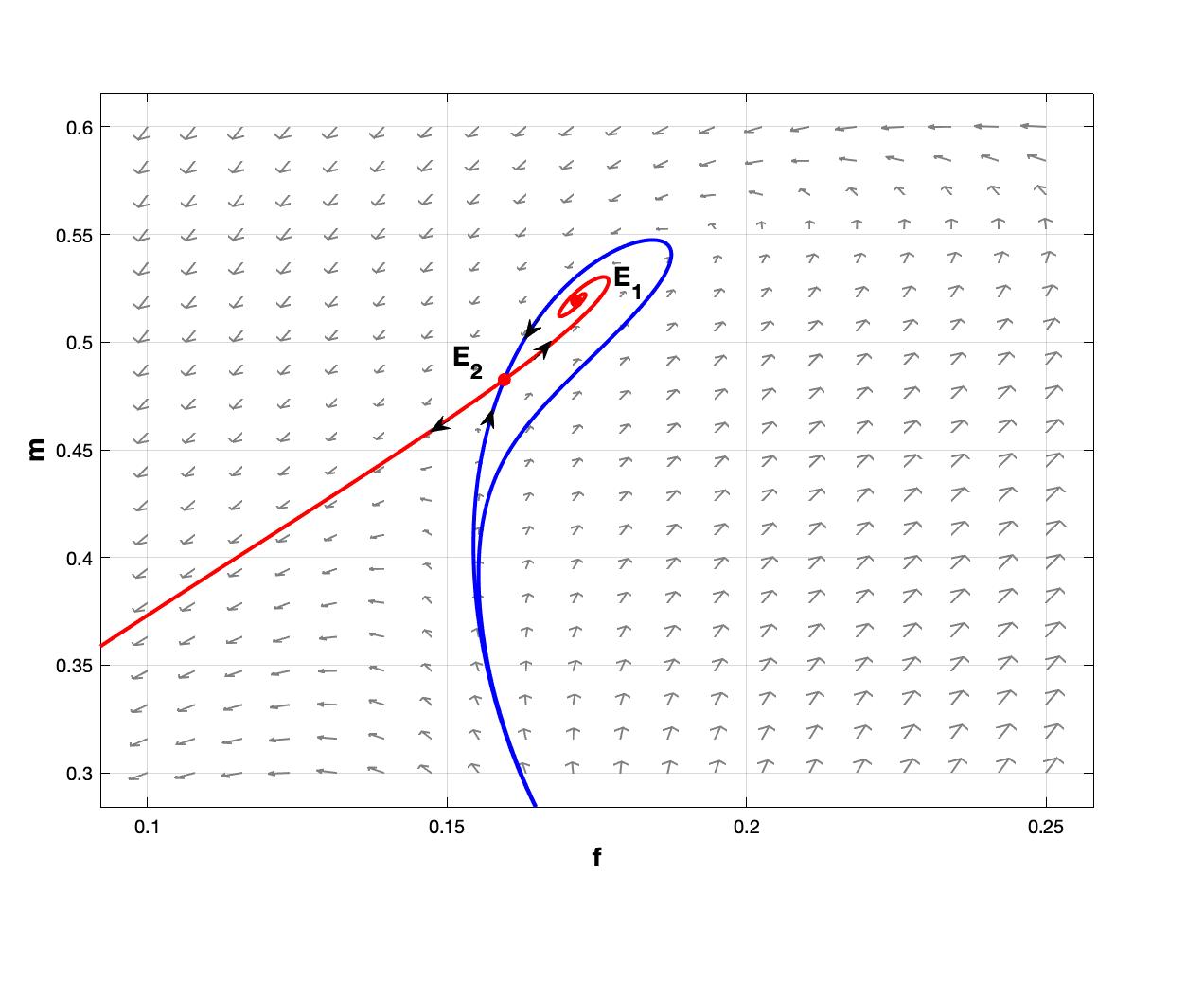} \\
(a) $s=0.56$&(b) $s=0.58$ &(c) $s=0.59$
\end{tabular}
\caption{Here we see the dynamics of a weak Allee effect in place. We notice in (a) and (b) the phase is split into two regions by the separatrix/stable manifold of the saddle $E_{2}$, with a large proportion of the initial data going to the recovery state $E_{1}$, which is a spiral sink. (c) shows this can change, as we vary the stocking parameter $s$ - the unstable manifold of $E_{2}$ can loop around $E_1$, and eventually lead to a homoclinic orbit. The other parameter values are $r=0.5, \alpha=90, h=0.24$}
\label{fig:nwa}
\end{figure}

\begin{figure}
   \centering
     \hspace*{-1.cm}
\begin{tabular}{ccc}
\includegraphics[scale=0.13]{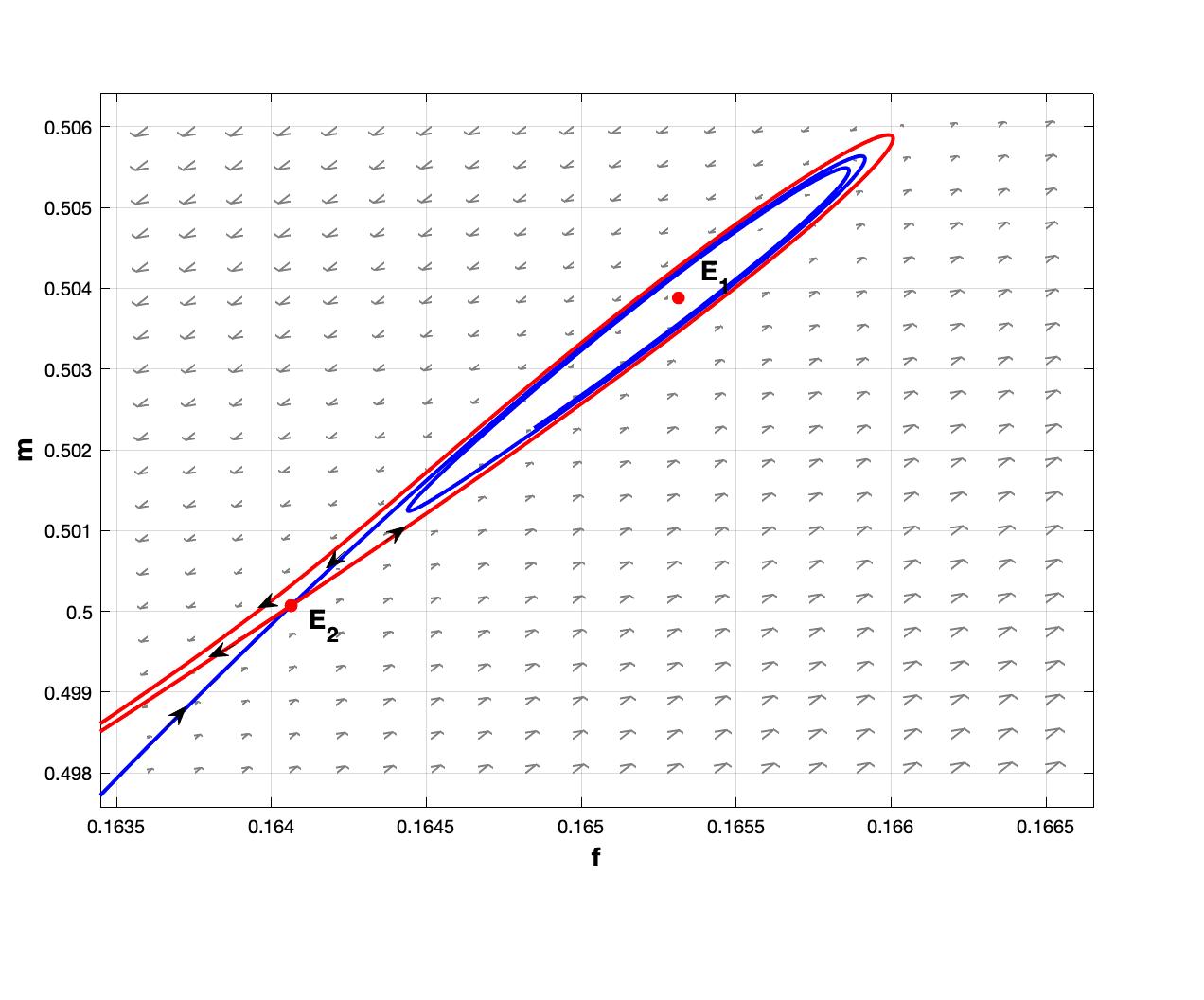}&
\includegraphics[scale=0.13]{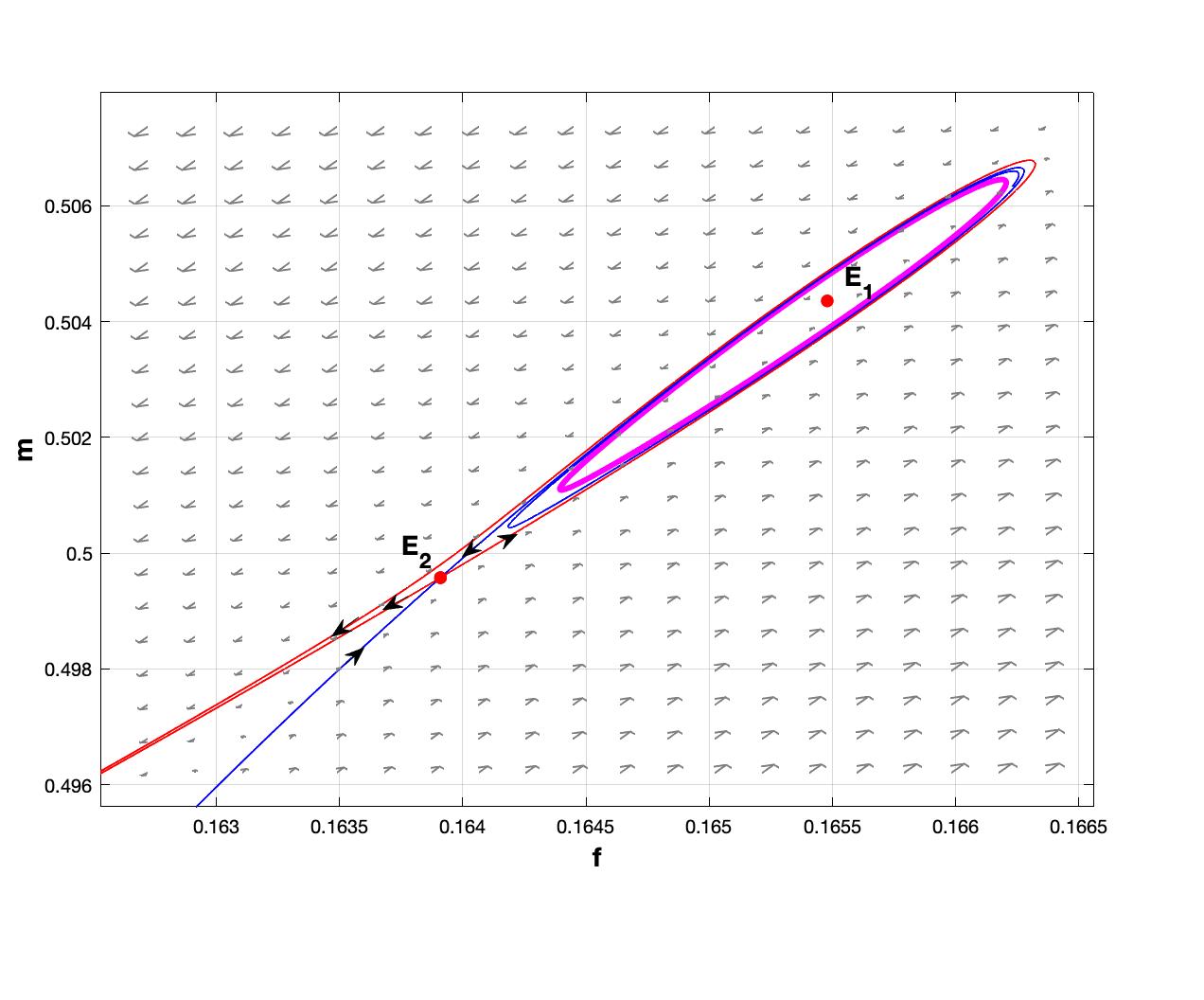} &
\includegraphics[scale=0.13]{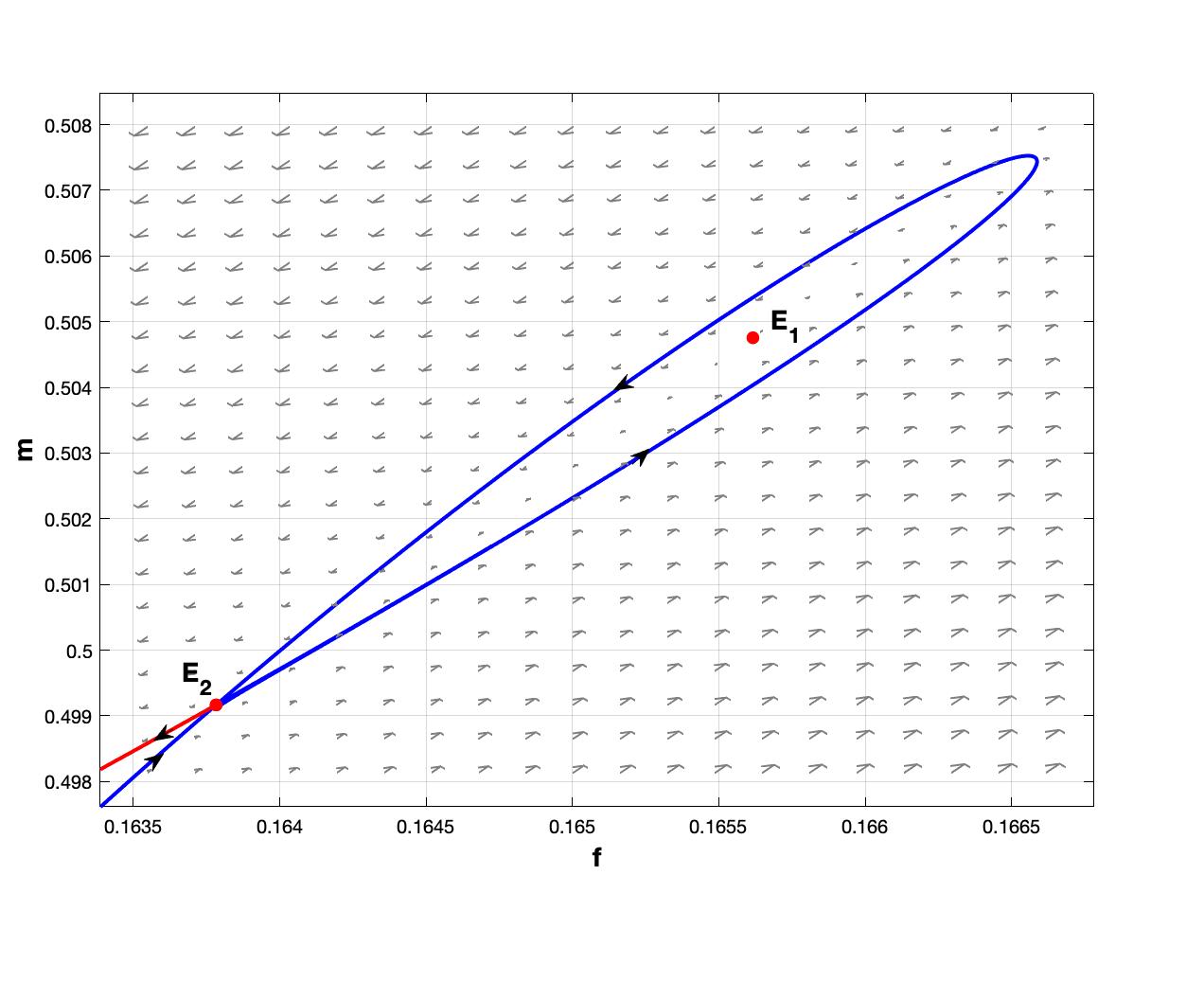} \\
(a) $s=0.59318$&(b) $s=0.59316$ &(c) $s=0.59314$
\end{tabular}
\caption{Phase space showing the path to the occurrence of a Homoclinic orbit for different stocking $s$ values. The stable and unstable manifolds are colored blue and red respectively. The unstable limit cycle is colored  pink. The parameter values are $r=0.5, \alpha=90, h=0.24$. $E_1$ is spiral source at $s = 0.59318$, it gains stability through a subcritical Hopf bifurcation, resulting in the occurrence of an unstable limit cycle at $s=0.59316$. The limit cycle grows at $s$ is decreased further, and at $s=0.59314$ it collides with the stable and unstable manifold and the saddle $E_2$ to form a homoclinic orbit}
\label{AfterHomoclinic}
\end{figure}

\subsection{The case of strong Allee effect}
Consider the scaled mating model without female harvesting and male stocking, but with a strong Allee effect
 \begin{equation}
\label{eq: strongAllee}
\begin{split}
\dfrac{d f}{d t}  =& r \alpha m (f-A)(1-f-m)-f\\
\dfrac{d m}{d t} =&r \alpha m (f-A)(1-f-m)-m ,
\end{split}
\end{equation}
where $f, m$ are the female and male population size and $A$ is the Allee threshold. Our aim is to investigate the extinction boundary as the Allee threshold $A$ is varied. What we notice is that for small Allee thresholds, $A < \approx 0.3$, the extinction boundary is non-hyperbolic again, see Fig. \ref{fig:Bendings}. However, it returns to a hyperbolic shape once the Allee threshold is increased past 0.7.


\begin{figure}[ht!]
   \centering
     \hspace*{-1.cm}
\begin{tabular}{ccc}
\includegraphics[scale=0.25]{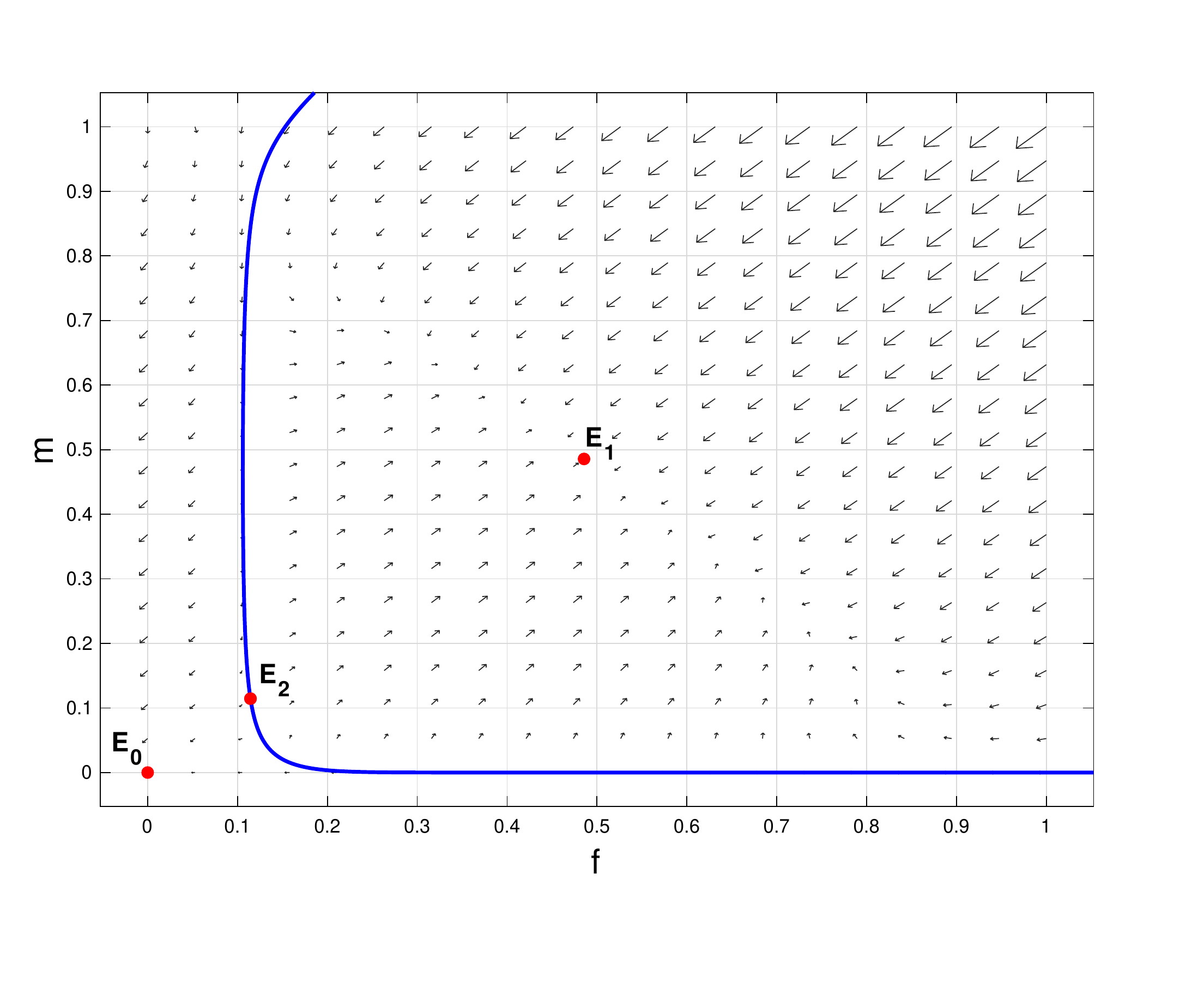}&
\includegraphics[scale=0.25]{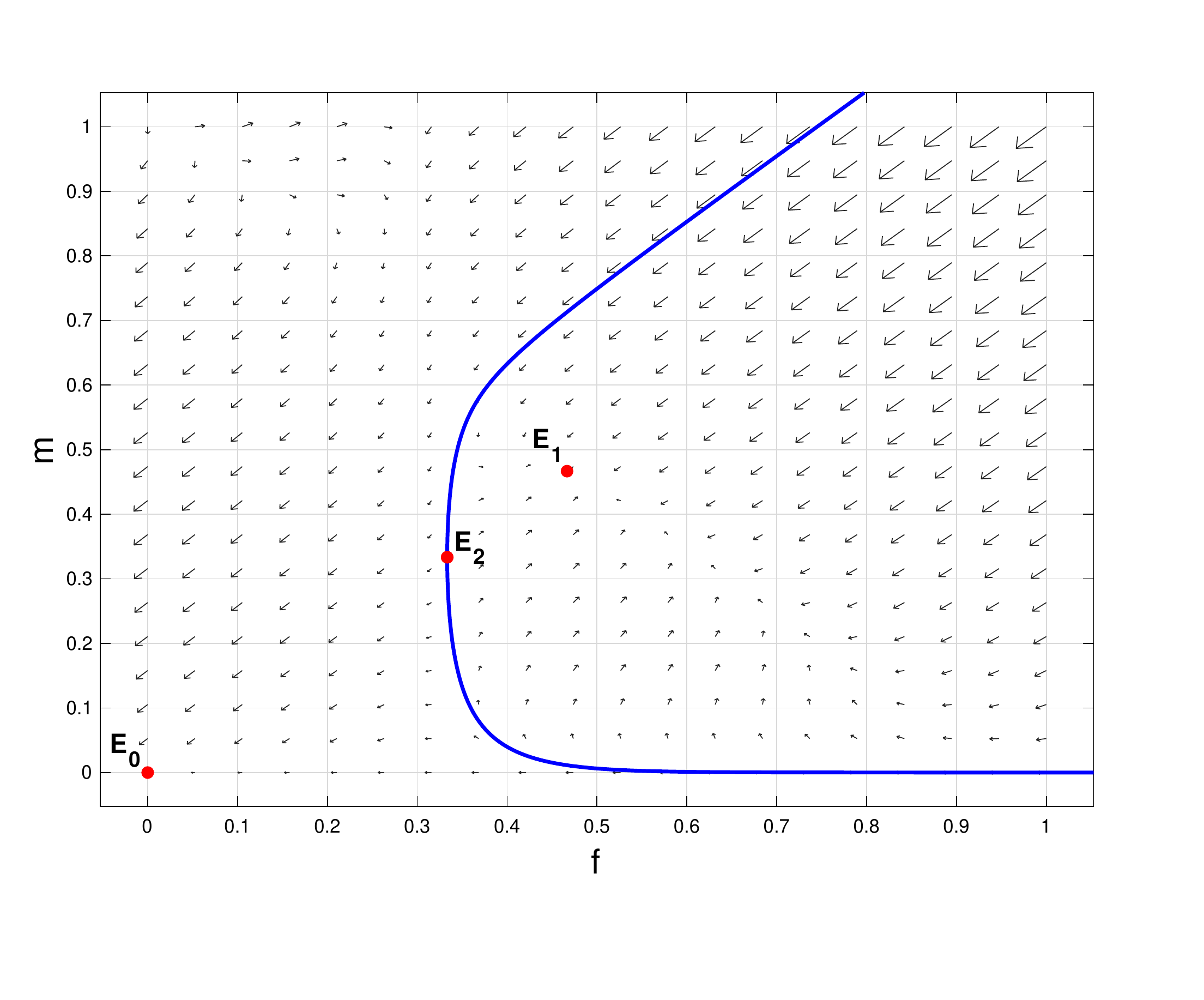} &
\includegraphics[scale=0.25]{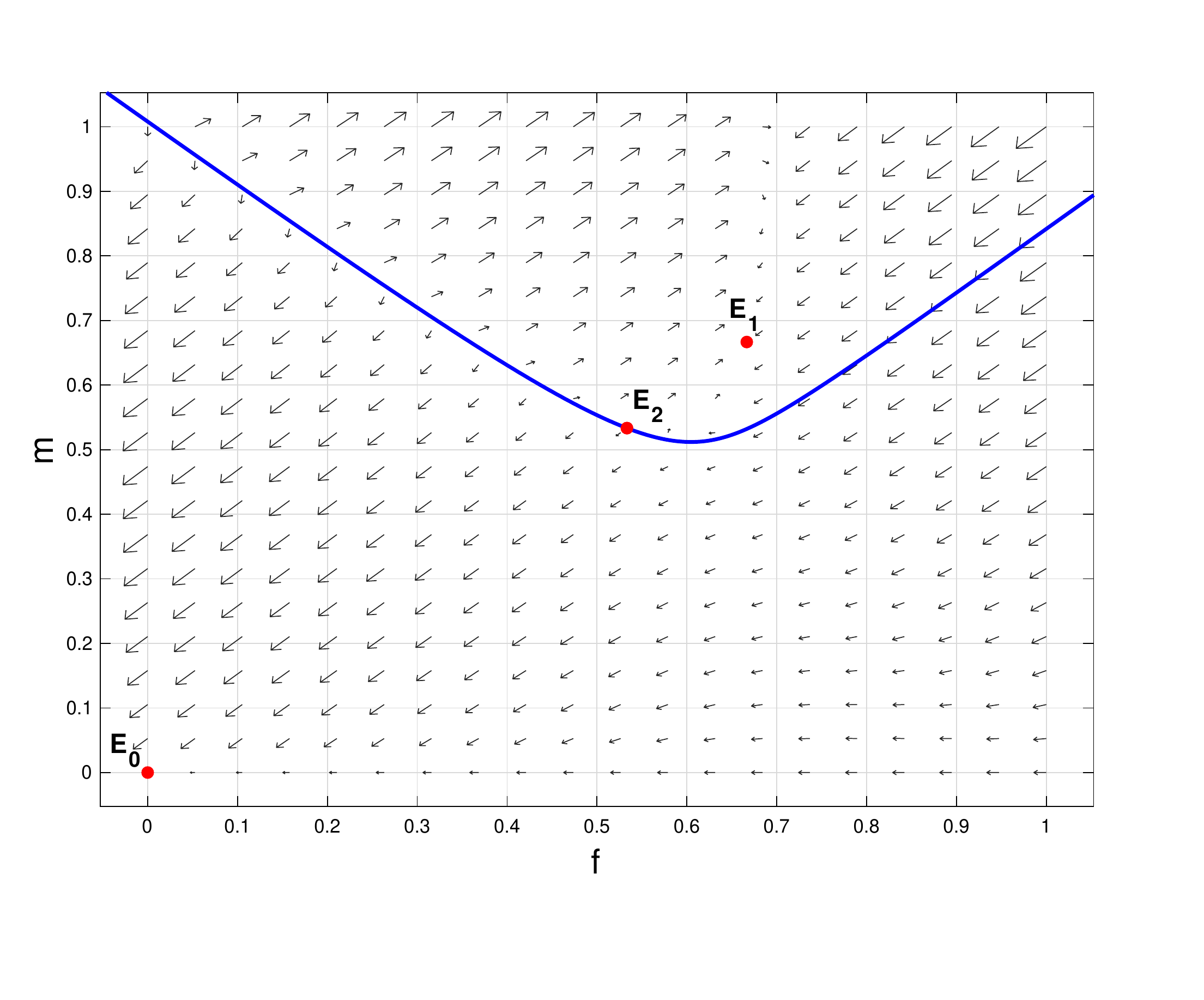} \\
(a) $A=0.1$&(b) $A=0.3$ &(c) $A=0.7$
\end{tabular}
 \caption{Here we observe the dynamics with a strong Allee effect in place via \eqref{eq: strongAllee}. The phase space of scaled mating model split into two regions by separatrix with different Allee thresholds. One can see that the extinction boundary bends, and is thus non-hyperbolic. Although a complete loop is not formed such as Fig. \ref{fig:nwa}. Also note increasing the Allee threshold (when a strong Allee effect is present) typically results in initial data from a larger region of the phase space, going to extinction. Here we notice an opposite effect, wherein increasing the Allee threshold from 0.3 to 0.7, causes a larger portion of the initial data to go to the recovery state.
 The parameter values $r=0.5$  and $\alpha=180$ are used in the simulation}
\label{fig:Bendings}
\end{figure}
\newpage
\section{Turing Instability}
Species disperse due to various factors such as search for food, mates, to look for resources as well as to avoid predators or competitors \cite{JDMurray}. Thus there is a rich history of spatially explicit models in mathematical biology \cite{O13}. With a spatially explicit model, just as with an ODE model, one can analyze the steady states. Herein, various additional dynamics are possible - in particular the states might not be homogenous in space. This could be indicative of the populations spreading and/or living in varying densities in a spatial domain. One might see the phenomenon of Turing instability per se \cite{JDMurray}, whereby large differences in the diffusion coefficients can destabilize the spatially homogenous steady state.
In this section, we explore the possibility of a Turing instability in systems (\ref{eq:3})  and (\ref{eq:5}) by including a spatial component.

\subsection{Spatially explicit FHMS}
We consider the system,
 \begin{equation}
\label{eq: NoTuringNoAllee}
\begin{split}
\dfrac{\partial f}{\partial t}  =& d_1 \Delta f + r \alpha f m \Big(1-(f+m) \Big)-(1+h)f \\
\dfrac{\partial m}{\partial t} =& d_2 \Delta m + (1-r) \alpha f m \Big(1-(f+m) \Big)+(s-1)m,
\end{split}
\end{equation}
where $f(x,t), m(x,t)$ are functions of both the spatial variable $x$ and time $t$. $\Delta$ is the standard Laplacian operator, representing the diffusion of the species in space. We prescribe Neumann boundary conditions $f_x=0, m_x=0$ which represents no flux of the species into or out of the spatial domain.
We also impose positive initial conditions $f(x,0)>0, m(x,0)>0$.  The positive constants $d_1$ and $d_2$ are the diffusion coefficients. We choose a suitable domain $\Omega= [0 ,\pi]$ for our numerical simulations.

\subsection{Spatially explicit FHMS with weak Allee effect}
Consider the system
 \begin{equation}
\label{eq: TuringAllee}
\begin{split}
\dfrac{\partial f}{\partial t}  =& d_{11} \Delta f + r \alpha f^2 m \Big(1-(f+m) \Big)-(1+h)f \\
\dfrac{\partial m}{\partial t} =& d_{22} \Delta m + (1-r) \alpha f^2 m \Big(1-(f+m) \Big)+(s-1)m,
\end{split}
\end{equation}
subject to Neumann boundary conditions $f_x=0, m_x=0$ and positive initial conditions $f(x,0)>0, m(x,0)>0$. The positive constants $d_{11}$ and $d_{22}$ are the diffusion coefficients. Similar to system (\ref{eq: NoTuringNoAllee}), we choose the domain $\Omega= [0 ,\pi]$.

\begin{theorem}.(Turing instability condition) Let $(f^*,m^*)$ be a non-trivial positive homogeneous steady state. For any given set of parameters, if the Jacobian $J^*_i=\begin{pmatrix}
  J_{11} & J_{12}\\
J_{21}& J_{22}
\end{pmatrix}$
of the reaction terms evaluated at $(f^*,m^*)$ and the diffusion constants $d_1, d_2$ satisfy
\begin{eqnarray}
\label{COND1} J_{11}+J_{22} < 0,\\
\label{COND2} J_{11}J_{22}-J_{21}J_{12} > 0,\\
\label{COND3} d_2J_{11}+d_1J_{22} > 0,\\
\label{COND4} (d_2J_{11}+d_1J_{22} )^2 - 4d_1d_2(J_{11}J_{22}-J_{21}J_{12} )> 0,
\end{eqnarray}
then in the absence of diffusion, $(f^*,m^*)$  is linearly stable and linearly unstable in the presence of diffusion.
\end{theorem}
We refer the reader to \citep{JDMurray} for a more detailed derivation of the conditions necessary and sufficient for the occurrence of Turing instability.
\begin{theorem}. (Necessary condition for Turing instability)
A necessary condition for the occurrence of Turing instability is  either $J_{11}<0, J_{22}>0$ or $J_{11}>0, J_{22}<0$ of the Jacobian $J^*_i.$
\end{theorem}

\begin{lemma}.
\label{lem:nt1}
System (\ref{eq: NoTuringNoAllee}) does not exhibit Turing instability.
\end{lemma}

\begin{proof}.
Clearly, $J_{11} < 0,J_{22} < 0$ in the Jacobian $J^*_i $ of the reaction terms in system (\ref{eq: NoTuringNoAllee}) (that is, the Jacobian of system (\ref{eq:3}) seen in Appendix \ref{AppA}) and hence does not meet the necessary condition for the occurrence of Turing instability. Hence proof.
\end{proof}

In this section, we use the following set of parameters for all numerical experiments:
\begin{equation}\label{eq:Param_Param}
r=0.5, \alpha=90,h=0.24, s=0.5931.
\end{equation}

\begin{theorem}.
\label{thm:t1}
System (\ref{eq: TuringAllee}) \textbf{exhibits} Turing instabilty.
\end{theorem}

\begin{proof}.
Consider the parameter set in Eq.(\ref{eq:Param_Param}). With these values,  a homogeneous steady state of system (\ref{eq: TuringAllee})  is $(f^*,m^*)=(0.165847, 0.505407 )$. The Jacobian evaluated at $(f^*,m^*)$  yields $J^*_i =\begin{pmatrix}
  0.614441 & -0.218659\\
1.85444 & -0.625559
\end{pmatrix}.$
\\ Clearly, $J_{11} = 0.614441> 0,J_{22} = -0.625559< 0$ of $J_i^*$ meets the necessary condition for Turing instability occurrence. The eigenvalues of $J_i^*$  are $\lambda_1=-0.00555889 + 0.145224 $i and $\lambda_2 = -0.00555889 - 0.145224  $i. The real parts of $\lambda_1, \lambda_2$  are both negative and hence the steady state $(f^*,m^*)$  is locally stable. Simple calculations show that, $J_{11}+J_{22} = -0.011118 < 0, J_{11}J_{22}-J_{21}J_{12} = 0.0211209 > 0, d_2J_{11}+d_1J_{22} = 0.0614 > 0$ and $(d_2J_{11}+d_1J_{22} )^2 - 4d_1d_2(J_{11}J_{22}-J_{21}J_{12} ) =  0.00376685 > 0.$ All conditions for the occurrence of Turing instability have been met. Hence proof.
\end{proof}

We define a small perturbation around the positive homogeneous steady state as
\begin{equation}\label{eq:perturb}
\begin{split}
f=& f^* + \alpha_1 \sin^2(nx)\\
m=&m^* + \alpha_2 \sin^2(nx),
\end{split}
\end{equation}
where $\alpha_1, \alpha_2, n \in \mathbf{R}$.

\begin{figure}[ht!]
\centering
 \hspace*{-1.5cm}
\begin{tabular}{cc}
\includegraphics[scale=0.38]{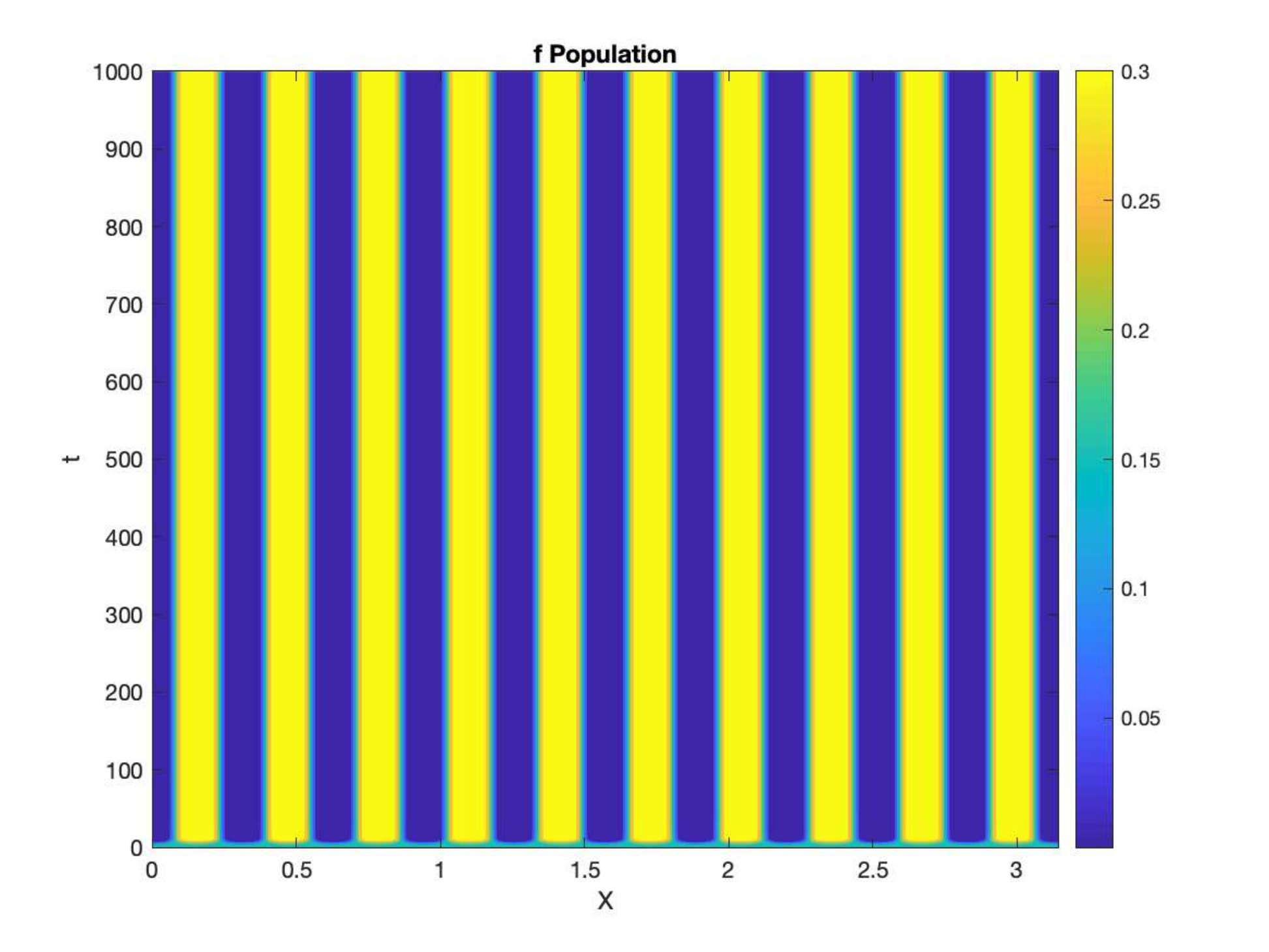} &
\includegraphics[scale=0.38]{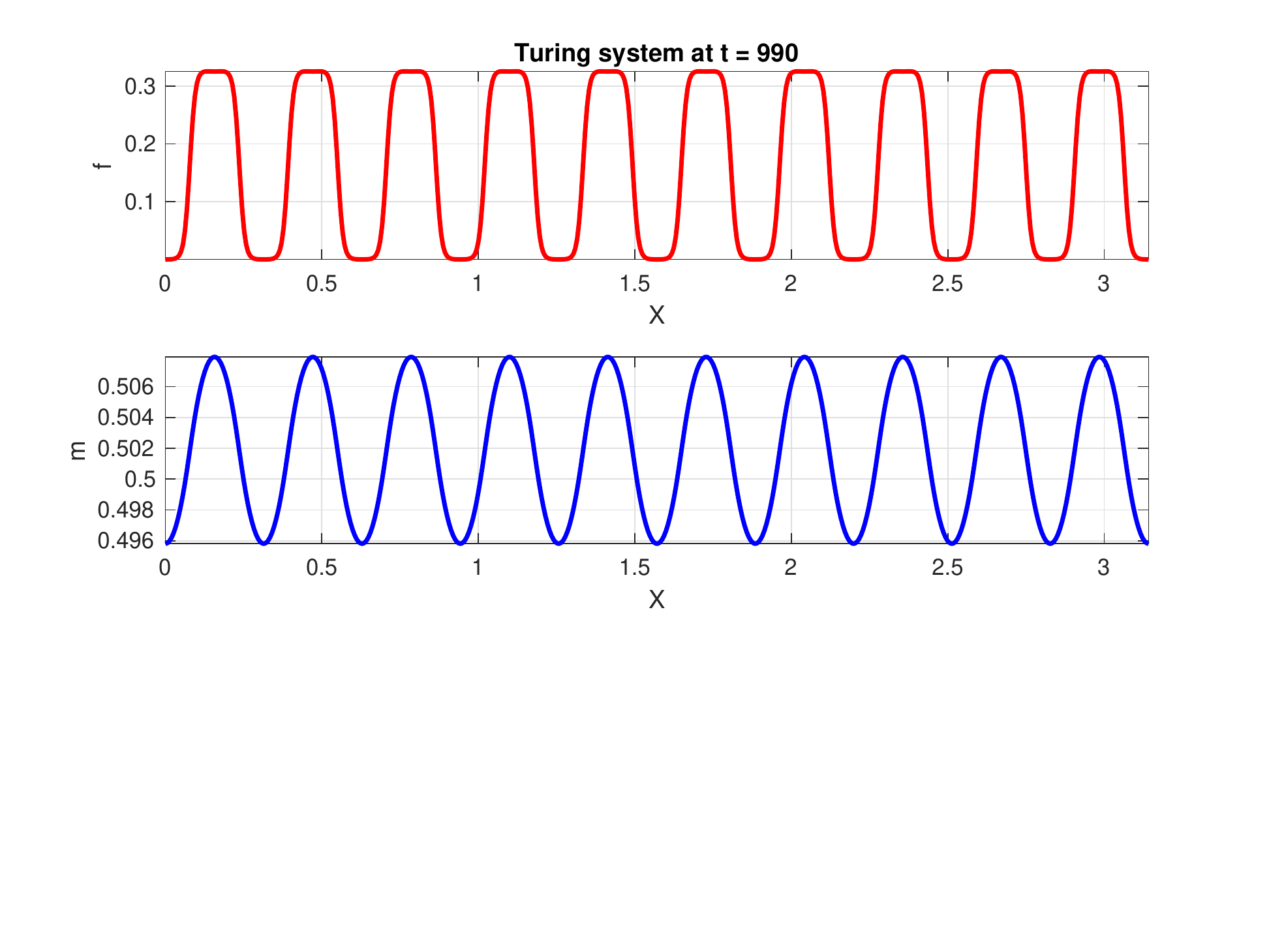} \\
(a) & (b) \\
\end{tabular}
\caption{Turing instability in FHMS with weak Allee in system (\ref{eq: TuringAllee}). (a) Turing patterns in female population. (b) Spatial profile of female and male populations after simulation is run longer ($t=1000$). The diffusion constants are $d_{11}=0.0001,d_{22}= 0.1$. We give the spatially positive homogeneous steady state a small perturbation as in Eq.(\ref{eq:perturb}) where $\alpha_1=0.01$, $\alpha_2=0.01$ and $n=10$}
\label{fig:Turing1}
\end{figure}

\section{Discussion and Conclusion}

In this work, we investigate models that mimic the TYC strategy, such as the FHMS models initiated in \citep{L18,L19}. These are critical research directions in biocontrol, given that supermale production has been successfully achieved by only one group in the US \citep{Schill17,Set16}, and the feminised supermale is still not in existence. Furthermore, given strong evidence of weak Allee effects at low population densities, it is important to incorporate them into models of biocontrol, particularly in those which attempt to reduce the female population density - as a means of driving the overall population to extinction. Here, the weak Allee effect is apt, because it works on positive density dependence, and populations still grow at low densities, only slower. So it is truly ``weaker" than a strong Allee effect, which would imply negative density dependence below the Allee threshold.

The interplay of the harvesting and stocking with the weak Allee effect is important for bi-stability. To elucidate, let us focus on Fig. \ref{fig03}. Herein in (a), there is no weak Allee effect. At a harvesting level of $h=0.5$, one can increase stocking unboundedly, but the system will remain bi-stable. However if we observe (b), where there is a weak Allee effect, at $h=0.5$, if one increases stocking above $s=0.25$, we have only mono stability, and all initial conditions go to the extinction state. This raises the question of what exactly is the extinction boundary in the FHMS system vs the one in the FHMS system with weak Allee effect.

To this end, we have derived various results with important implications for control. We show that the FHMS model with weak Allee effect in conjunction with harvesting and stocking, can produce a non-hyperbolic extinction boundary, see Fig. \ref{fig:nwa}. This, to the best of our knowledge, is the first example of such a boundary in 2 species structured, mating models. The implication for control (for parameter choices say in Fig. \ref{fig:nwa}) is that in situations, where we raise the stocking by just $1\%$, we can get essentially any initial condition, driven to the extinction state. It is important to note, that the weak Allee effect, and harvesting and stocking pressures are not exclusively responsible for a non-hyperbolic shaped  extinction boundary - one can see this shape or ``bending" even when a strong Allee effect is present, see Fig. \ref{fig:Bendings}. Although a complete loop is not formed such as Fig. \ref{fig:nwa}. This result is counter intuitive, in that typically increasing the Allee threshold (when a strong Allee effect is present) results in initial data from a larger region of the phase space, going to extinction - but the opposite is observed here, see Fig. \ref{fig:Bendings}.

Our goal to introduce the results with strong Allee effect here is purely motivational and as a conduit for future investigations. We refrain from the mathematical analysis of the strong Allee effect presently, and relegate it for detailed future work. Note current works in this vein, that consider strong Allee effects in TYC context \citep{JB20,BPL20}, only find hyperbolic extinction boundaries, but this is probably due to the parameter choices of simulations therein.
This begets the question of how in general does the extinction boundary look in mating models, with the inclusion of Allee effects, which seems to be an open problem in ecology \citep{B04}. For systems that are monotone, say of competitive or cooperative type, there is a large body of results, see \citep{S04,SJ09,ST01,SH06,J04} and the references within, that point to the hyperbolic shape of the threshold manifold. However, the systems we consider are non-monotone (easily observed by checking the signs of the off-diagonal terms of the jacobian matrix) - and the threshold boundary for such systems is less investigated.

Our other result concerns the stopping of harvesting or stocking, at various finite times, and initial densities of the wild-type, so as to still yield extinction. This is elucidated via Figs. \ref{fig04} - \ref{fig05}. Lastly, we would like to comment on our results concerning spatial/Turing instability. Here again, we see via Lemma \ref{lem:nt1} and Theorem \ref{thm:t1}, that the FHMS system cannot produce Turing patterns, whereas the FHMS system with weak Allee effect can. Thus the weak Allee effect can cause the population of males and females to spread patchily in space. This has been observed elsewhere in the literature as well \citep{PQB16,PQ16}.

\vspace{0.5cm}
\noindent {\bf{Acknowledgements}}\\
ET and RP would like to acknowledge valuable support from the National Science Foundation via DMS 1839993. MB would like to acknowledge valuable support from the National Science Foundation via DMS 1715044. JB is supported by the grants from Science and Engineering Research Board (SERB), Govt. of India (File No. TAR/2018/000283).

 \section*{Conflict of interest}

 The authors declare that they have no conflict of interest.



\begin{appendices}
\section{Stability analysis for FHMS model without weak Allee effect}\label{AppA}
The system (\ref{eq:3}) has the nullclines $F_i=0$ $(i=1,2)$. Solving these nullclines yields the following equilibria:

$(i)$ Invasive fish-free equilibrium $E_0=\left(0,0\right)$ exists always and is locally asymptotically stable (as $s<1$).

$(ii)$ coexistence equilibria $E^*_{i}=\left(f^*_i,\mu f^*_i\right)$, where $\mu=\frac{(1-r)(1+h)}{r(1-s)}$ and
$$f^*_i=\frac{1}{2(1+\mu)}\left\{1\pm\sqrt{1-\frac{4(1+\mu)(1+h)}{r\alpha\mu}}\right\}, \;\; (i=1,2).$$
The following lemma gives the conditions for existence of the unique interior equilibrium $E^*_i$ $(i=1,2)$:
\begin{lemma}\label{lem:A}.
The interior equilibrium $E^*_i$ of the system \eqref{eq:3}  exists if either one of the following two conditions holds:\\
$(i)$ $s\geq s^*$ and $h_*<h<h^*$;\\
$(ii)$ $0\le h\le h_*$ and $\alpha >\frac{4}{r(1-r)}$, where $s^*=\frac{1}{r}+\left(\frac{1-r}{4}\right)\left(\frac{4h}{r}-\alpha\right)$, $h_*=\frac{r}{4}\left\{\alpha-\frac{4}{r(1-r)}\right\}$ and $h^*=\frac{r}{4}\left(\alpha-\frac{4}{r}\right)$.
\end{lemma}
The linearized system of \eqref{eq:3} about an equilibrium $\hat{E}$ is given by $\frac{dX}{dt}=J(\hat{E})X$, where $X=\left(f \;\; m\right)^T$ and $J(\hat{E})$ is the Jacobian matrix of the system \eqref{eq:3} evaluated at $\hat{E}$. We analyze the stability of system \eqref{eq:3} by using eigenvalue
analysis of the Jacobian matrix evaluated at the appropriate equilibrium. At $E_0$, the eigenvalues of the Jacobian matrix of the system \eqref{eq:3} are $s-1$ and $-(1+h)$. Since $0\le s<1$, all the eigenvalues of the Jacobian matrix $J(E_0)$ are negative. This gives the following lemma:

\begin{lemma}.
The invasive fish-free equilibrium $E_0$ is always locally asymptotically stable.
\end{lemma}

The Jacobian of the system \eqref{eq:3} evaluated at $E^*_i$ is given by
$$J^*_i=\begin{pmatrix}
  -r\alpha \mu f^{*2}_i & r\alpha f^*_i(1-f^*_i-2\mu f^*_i)\\
  (1-r)\alpha \mu f^*_i(1-2f^*_i-\mu f^*_i) & -(1-r)\alpha \mu f^{*2}_i
\end{pmatrix}.$$
We have $\mbox{Tr}(J^*_i)=-\alpha\mu f^{*2}_i<0$ and $\mbox{Det}(J^*_i)=-r(1-r)\alpha^2 \mu f^{*2}_i \left\{1-2(1+\mu)f^*_i\right\}\left\{1-(1+\mu)f^*_i\right\}$.\\ Therefore, the system \eqref{eq:3} is locally asymptotically stable at $E^*_i$ if and only if $\mbox{Det}(J^*_i)>0$. This gives the following lemma:
\begin{lemma}.
Assume that the conditions of Lemma \ref{lem:A} are satisfied. If $\frac{r(1-s)}{2\{1+h-r(h+s)\}}<f^*_i<\frac{r(1-s)}{1+h-r(h+s)}$ holds, then the system \eqref{eq:3} is locally asymptotically stable at $E^*_i$.
\end{lemma}

\section{FHMS with weak allee}
\subsection{Proof of boundedness of FHMS model with weak Allee effect \eqref{eq:5} }\label{AppB}

\begin{proof}.
Case $(i)$: \\
Let $0\le f(0)<1$. If possible, let there exists $t>0$ such that $f(t)\geq 1$. We define $t_0=\min\left\{t:f(t)\geq 1\right\}$.
Then $f(t_0)=1$ and $f(t)<1$ for $0\le t<t_0$. Now, we have $f^\prime(t_0)=-r\alpha m^2(t_0)-(1+h)<0$.  By the continuity of $f^\prime(t)$, there exists $\delta>0$ such that $f^\prime (t)<0$ for all $t\in(t_0-\delta,\;\;t_0+\delta)$.

Let $t_1=t_0-\frac{\delta}{2}$. Then $t_0-\delta<t_1<t_0<t_0+\delta$. Since $f(t)$ is strictly decreasing function for all $t\in (t_0-\delta,\;\;t_0+\delta)$, we have $f(t_1)>f(t_0)=1$ which contradicts to the definition of $t_0$. Therefore, we can conclude that $f(t)\geq 1 $ cannot be true for any $t>0$ when $0\le f(0)<1$.\\

Case $(ii)$:\\
 Let $f(0)\geq 1$. We first assume that $f(0)=1$. Then $f^\prime(0)=-r\alpha m^2(0)-(1+h)<0$ and so, there exists $\delta_0>$ such that $f^\prime(t)<0$ for all $t\in [0,\;\;\delta_0)$. Therefore, for all $t>\delta_0$, we have $f(t)<f(\delta_0)<f(0)=1$.

Next assume that $f(0)>1$. Then $f^\prime(0)<0$ implies there exists $\delta_1>0$ such that $f^\prime(t)<0$ for all $t\in [0,\;\;\delta_1).$ Since $f^\prime(t)$ is strictly decreasing for all $t\in[0,\;\;\delta_1)$, it follows that $f(\delta_1)<f(0)$. Suppose that there exists $t>\delta_1$ such that $f(t)>f(0)$. Let $t_2>\delta_1$ be defined by $t_2=\min\left\{t>\delta_1: f(t)>f(0)\right\}$. Then $f(t_2)>f(0)>1$ implies $f^\prime(t_2)=r\alpha f^2(t_2)m(t_2)\left(1-f(t_2)-m(t_2)\right)-(1+h)f(t_2)<0$ and so, there exists $\delta_2>0$ such that $f^\prime(t)<0$ for all $t\in(t_2-\delta_2,\;\;t_2+\delta_2)$.

One now sees that $f^\prime\left(t_2-\frac{\delta_2}{2}\right)<0$ from which it follows that $f\left(t_2-\frac{\delta_2}{2}\right)>f(t_2)>f(0)>1$, contradicting to the definition of $t_2$. Therefore, for $f(0)\geq 1$, there cannot exist $t>t_2$ such that $f(t)>f(0)$. Hence, all the solutions of the system \eqref{eq:5} are contained in some bounded subset in the first quadrant of the $fm$-plane.
\end{proof}

\subsection{Proof of saddle-node bifurcation in FHMS with weak Allee effect \eqref{eq:5} }
\label{AppB1}

\begin{proof}.
Solving $G(g)=0=G^\prime(f)$ we see that the equation $G(f)=0$ has a double root $\frac{2}{3(1+\mu)}$ satisfying $G^{\prime\prime}\left(\frac{2}{3(1+\mu)}\right)=\frac{4}{(1+\mu)}\neq 0$. The two nontrivial nullclines $\phi_i(f,m)=0$ $(i=1,2)$ intersect at the instantaneous interior equilibrium $E^*=\left(\frac{2}{3(1+\mu)},\; \frac{2\mu}{3(1+\mu)}\right)$. At $E^*$, the slopes of $\phi_i(f,m)=0$ $(i=1,2)$ are equal and so $\dfrac{\bar{F}_{1_f}}{\bar{F}_{1_m}}=\dfrac{\bar{F}_{2_f}}{\bar{F}_{2_m}}$, which gives $\mbox{Det}(J^*)=0$. Solving $\mbox{Det}(J^*)=0$, the critical value of $h$, say, $h=h^*$ can be obtained. At at $h=h^*$, if $\mbox{Tr}(J^*)\neq 0$, then the Jacobian of the system \eqref{eq:5} has a simple zero eigenvalue.

Let $\bar{F}(f,m;s)=\left(\bar{F}_1\;\; \bar{F}_2 \right)^T$ and $V^*$ and $W^*$ are eigenvectors corresponding to the zero eigenvalue for $J^*{_|{_{h=h^*}}}$ and $J^{*T}{_|{_{h=h^*}}}$ respectively. We obtain $\bar{F}_h(f,m;h)=\left(\frac{-2}{3(1+\mu)}, \; 0 \right)^T$, $U^*=\left(1 \;\; \mu \right)^T$ and $V^*=\left(1\;\; \frac{3r(1-\mu)}{4(1-r)\mu}\right)^T$ so that  $V^{*T} \bar{F}_{h}\left(E^*;h^*\right)=\frac{-2}{3(1+\mu)}$ and $V^{*T}(D\bar{F}_h)(U^*)=-1$.\\
Due to the complexity in the algebraic expressions involved, we will use numerical simulations to verify $V^{*T}\left[D^2\bar{F}\left(E^*;h^*\right)(U^*,U^*)\right]\neq 0$. Under these conditions, the system \eqref{eq:5} satisfies Sotomayor's theorem for a saddle-node bifurcation at $E^*$ when $h$ crosses $h^*$. This gives the following lemma.

Keeping all parameters fixed and varying the harvesting parameter $h$ we observe that the coexisting equilibria $E^*_i$ $(i=1,2)$ collide to each other, generating a unique instantaneous interior equilibrium. From Fig. \ref{fig02}$(b)$, it is observed that the interior equilibria $E^*_i$ $(i=1,2)$ cease to exist when $h$ is increased through $h^*=0.4$. \\
At $h=0.4$, we have $E^*=(0.2,\;0.47)$ $(i=1,2)$ and
$$J^*_i=\begin{pmatrix}
 0.56 & -0.24\\
 1.96 & -0.84
\end{pmatrix}$$ has a simple zero eigenvalue. Also, we obtain $U^*=\left(1 \;\; 2.33\right)^T$, $V^*=\left(1\;\; -0.286\right)^T$, $V^{*T} \bar{F}_{h}\left(E^*;h^*\right)=-0.2$, $V^{*T}(D\bar{F}_h)(U^*)=-1$ and $V^{*T}\left[D^2\bar{F}\left(E^*;h^*\right)(U^*,U^*)\right]=-29.997$, satisfying the conditions of a saddle-node bifurcation at $E^*$ when $h$ crosses $h=h^*$.
\end{proof}

\subsection{Proof of Hopf bifurcation in FHMS system with weak Allee effect \eqref{eq:5} }
\label{AppB3}
%
%

\begin{proof}.
We consider the following parameter set $r=0.5, \alpha=90, h=0.24 , s =0.59317665$. Then $E_1^*=(0.16534,0.50397)$ is an interior equilibrium. Evaluating the Jacobian of system (\ref{eq:5}) at $E_1^*$, we obtain
\begin{equation}\label{eq:Jacobi}
J_i^*=
\begin{pmatrix}
0.62 &  -0.21318\\
1.86 & -0.62
\end{pmatrix}.
\end{equation}
The corresponding eigenvalues are given as $\lambda_{1,2}=\pm 0.11005$i.  Clearly, the trace and determinant of Eqn. (\ref{eq:Jacobi}) is $0$ and $0.01211 > 0$ respectively. Now, referencing the Jacobian of system (\ref{eq:5}) again, let $\mu=\dfrac{\lambda}{(1-s)}$ where $\lambda=\dfrac{(1-r)(1+h)}{r}$. Then
\begin{equation}\label{eq: HOPHOP}
\dfrac{d}{ds}(\text{\mbox{Tr}}(J_i^*)) =\dfrac{\lambda r \alpha f_i^{*2}}{(1-s)^2}\Bigg(1-f_i^*\Bigg(1+\dfrac{1}{r} +\dfrac{2\lambda}{(1-s)} \Bigg)      \Bigg).
\end{equation}
Now, $\dfrac{d}{ds}(\text{\mbox{Tr}}(J_i^*))\big |_{E_1^*} = -4.6446\neq 0$ at s=0.59317665. Hence, the FHMS model with weak Allee effect  undergoes a Hopf bifurcation with respect to the bifurcation parameter $s=s_{hf}=0.59317665.$

We now vary the stocking rate $s$. As seen in Fig. \ref{AfterHomoclinic}, when $s=0.59318$, the interior equilibrium $E_1^*=(0.16531,0.50388)$ is a spiral source. The eigenvalues associated to $E_1^*$ are $\lambda_1=0.00033529+0.10753 $i and $\lambda_2=0.00033529-0.10753 $i. When the stocking rate  is decreased to $s=0.59316$, $E_1^*=(0.16548,0.50436)$ gains stability and becomes a spiral sink. It's associated eigenvalues are $\lambda_1=-0.0014934+0.12056$i and $\lambda_2=-0.0014934-0.12056$i. Clearly, the pair of complex conjugate eigenvalues cross the imaginary axis of the complex plane when the stocking rate $s$ decreases from $s=0.59318$ to $s=0.59316$. This leads to a subcritical Hopf bifurcation which gives rise to an unstable limit cycle.
\end{proof}

\end{appendices}

\end{document}